\newcommand{\arXiv}{}

\ifdefined\arXiv
\documentclass{article}

\pdfoutput=1

\usepackage[margin=.8in]{geometry}

\usepackage{authblk}
\else

\documentclass{vldb}
\newcommand{\VLDB}{}
\fi

\usepackage{graphicx}
\usepackage{balance}  % for  \balance command ON LAST PAGE  (only there!)

\ifdefined\arXiv
\else
\vldbTitle{Stream Frequency Over Interval Queries}
\vldbAuthors{Ran Ben Basat, Roy Friedman, Rana Shahout}
\vldbDOI{https://doi.org/TBD}
\vldbVolume{12}
\vldbNumber{xxx}
\vldbYear{2019}
\fi

\usepackage{algpseudocode}
\usepackage{bm}
\usepackage{graphicx}
\usepackage{comment}
\usepackage{array}
\usepackage{float}
\ifdefined\VLDB
\newcommand{\subparagraph}{}
\fi
\let\proof\relax
\let\endproof\relax
\usepackage{amsthm}
\usepackage{amsmath}
\usepackage{amssymb}
\usepackage{algorithm}
\usepackage{subfig}
\usepackage{titlesec}
\usepackage{blindtext}
\usepackage{multirow}
\usepackage{yfonts}
\usepackage{xcolor}
\usepackage{url}
\usepackage{balance}
\allowdisplaybreaks
\usepackage[font=small,skip=7pt]{caption}

\ifdefined\VLDB
\newcommand*{\NINEPAGES}{}
\newcommand*{\VLDBREVISION}{}
\fi
\ifdefined\NINEPAGES
\else
\newcommand*{\EXTENDED}{}
\fi

\ifdefined\VLDB
\newcommand{\changed}[1]{\textcolor{blue}{#1}} 
\else
\newcommand{\changed}[1]{\textcolor{black}{#1}} 
\fi

\usepackage{pifont}% http://ctan.org/pkg/pifont
\newcommand*{\inlineequation}[2][]{%
	\begingroup
	\refstepcounter{equation}%
	\ifx\\#1\\%
	\else
	\label{#1}%
	\fi
	\relpenalty=10000 %
	\binoppenalty=10000 %
	\ensuremath{%
		#2%
	}%
	~\@eqnnum
	\endgroup
}
\algtext*{EndFunction}
\algtext*{EndIf}
\algtext*{EndFor}
\algtext*{EndWhile}
\algtext*{EndIf}

\newtheorem{theorem}{Theorem}
\newtheorem{definition}{Definition}
\newtheorem{corollary}{Corollary}
\newtheorem{lemma}{Lemma}
\newtheorem{observation}[theorem]{Observation}
\newenvironment{sproof}{%
	\proof}{\endproof}

\ifdefined\EXTENDED
\newcommand{\matrixCellWidth}{5.8cm}
\else
\newcommand{\matrixCellWidth}{5.5cm}
\fi
\newcommand{\tableDis}{d}
\renewcommand{\mid}{\ensuremath{:}}
\newcommand{\curFrameFreq}{\ensuremath{\mathit{cFreq}}}
\newcommand{\eps}{\epsilon}
\newcommand{\brackets}[1]{\ensuremath{\left[#1\right]}}
\newcommand{\BS}{\ensuremath{\mathit{BlockSize}}}
\newcommand{\Tables}{\ensuremath{\mathit{Tables}}}
\newcommand{\incTables}{\ensuremath{\mathit{incTables}}}
\newcommand{\incTable}{\ensuremath{\mathit{incTable}}}
\newcommand{\ghostTables}{\ensuremath{\mathit{ghostTables}}}
\newcommand{\OFFSET}{\ensuremath{\mathit{offset}}}
\newcommand{\blockSize}{\ensuremath{\mathit{s}}}
\newcommand{\IFText}{IntervalFrequency}
\newcommand{\IFTextD}{Interval-Frequency}
\newcommand{\TIFText}{TimeIntervalFrequency}
\newcommand{\IHHText}{IntervalHeavyHitters}
\newcommand{\SIText}{Interval}
\newcommand{\SIproblemParams}{n}
\newcommand{\SIProblem}{${\SIproblemParams}${-}{\SIText}}
\newcommand{\frameOffset}{\ensuremath{\mathit{fo}}}
\newcommand{\SSInstance}{\ensuremath{\mathit{SS}}}
\newcommand{\SIQText}{\SIText{}Query}
\newcommand{\IFQText}{\IFText{}Query}
\newcommand{\TIFQText}{\TIFText{}Query}
\newcommand{\IHHQText}{\IHHText{}Query}
\newcommand{\problemParams}{(W,\epsilon)}
\newcommand{\timeProblemParams}{(\Tau,R,\epsilon)}
\newcommand{\IFProblem}{${\problemParams}${-}{\IFText}}
\newcommand{\IFProblemD}{${\problemParams}${-}{\IFTextD}}
\newcommand{\ParameterizedIFProblem}[2]{${(#1,#2)}${-}{\IFText}}
\newcommand{\TIFProblem}{${\timeProblemParams}${-}{\TIFText}}
\newcommand{\IHHProblem}{${\problemParams}$\textbf{-}{\IHHText{}}}

\newcommand{\IFQGuarantee}[1][]{$$\xIntervalFrequency \le \xIntervalFrequencyEstimator \le \xIntervalFrequency + \epsError.#1$$}
\newcommand{\intervalHHThreshold}{\ensuremath{\theta\cdotpa{j-i}}}
\newcommand{\IHHDef}{\ensuremath{\set{x\in\mathcal U\mid \xIntervalFrequency\ge \intervalHHThreshold}}}
\newcommand{\xWindowFrequency}[1][w]{\ensuremath{f^{#1}_x}}
\newcommand{\xTimeWindowFrequency}[1][t]{\ensuremath{h^{#1}_x}}
\newcommand{\xWindowFrequencyEstimator}{\ensuremath{\widehat{\xWindowFrequency}}}
\newcommand{\xIntervalFrequency}{\ensuremath{f^{i,j}_x}}
\newcommand{\xTimeIntervalFrequency}{\ensuremath{h^{i,j}_x}}
\newcommand{\xTimeIntervalFrequencyEstimator}{\ensuremath{\widehat{\xTimeIntervalFrequency}}}
\newcommand{\xSetIntervalFrequency}{\ensuremath{g^{i, j}_x}}
\newcommand{\xSetWindowFrequency}[1][n]{\ensuremath{g^{#1}_x}}

\newcommand{\xParameterizedIntervalFrequency}[2]{\ensuremath{f^{#1,#2}_x}}
\newcommand{\xIntervalFrequencyEstimator}{\ensuremath{\widehat{\xIntervalFrequency}}}
\newcommand{\IntervalHH}[1][\theta]{\ensuremath{HH_{#1}^{i,j}}}
\newcommand{\IntervalHHEstimator}{\ensuremath{\widehat{\IntervalHH}}}
\newcommand{\epsError}[1][W]{\ensuremath{#1 \epsilon}}
\newcommand{\IFQ}{{\sc {\IFQText}}}
\newcommand{\SIQ}{{\sc {\SIQText}}}
\newcommand{\IHHQ}{{\sc {\IHHQText}}}
\newcommand{\BIFQ}{{\sc \textbf{\IFQText}}}
\newcommand{\BTIFQ}{{\sc \textbf{\TIFQText}}}
\newcommand{\BSIQ}{{\sc \textbf{\SIQText}}}
\newcommand{\BIHHQ}{{\sc \textbf{\IHHQText}}}
\newcommand{\set}[1]{\left\{#1\right\}}
\newcommand{\ceil}[1]{ \left\lceil{#1}\right\rceil}
\newcommand{\floor}[1]{ \left\lfloor{#1}\right\rfloor}
\newcommand{\angles}[1]{ \left\langle{#1}\right\rangle}
\newcommand{\logp}[1]{\log\parentheses{#1}}

\newcommand{\parentheses}[1]{ \left({#1}\right)}

\newcommand{\cdotpa}[1]{\cdot\parentheses{#1}}
\newcommand{\inc}[1]{$#1 \gets #1 + 1$}
\newcommand{\dec}[1]{$#1 \gets #1 - 1$}
\newcommand{\range}[2][0]{#1,1,\ldots,#2}
\newcommand{\frange}[1]{\set{\range{#1}}}

\newcommand{\oneOverE}{ \ensuremath{\eps^{-1}} }
\newcommand{\oneOverG}{ \frac{1}{\gamma} }

\newcommand{\window}{W}
\newcommand{\logw}{\log \window}

\newcommand{\weps}{\window\epsilon}

\newcommand{\numBlocks}{n}

\newcommand{\logn}{\ensuremath{\log\winSize}}

\newcommand{\winSize}{\ensuremath{n}}

\newcommand\Tau{\mathrm{T}}
\newcommand{\offset}{\ensuremath{\mathit{offset}}}
\newcommand{\RSS}{Rounded Space Saving}
\newcommand{\rss}{RSS}
\newcommand{\slack}{1+\gamma}

\newcommand{\frqEst}{$(\epsilon,M)$-{\sc Volume Estimation}}

\newcommand{\query}{{\sc Query$(x)$}}
\newcommand{\winQuery}{{\sc WinQuery$(x)$}}
\newcommand{\windowQueryText}{{\sc WinQuery}}
\newcommand{\windowQuery}{\windowQueryText$(x,w)$}

\newcommand{\SSSInstance}{rss}
\newcommand{\queueOfOverflows}{b}
\newcommand{\sumOfOverflows}{B}

\newcommand{\overflowIndicator}{u}

\renewcommand{\gamma}{\phi} 

\begin{document}

% ****************** TITLE ****************************************

%\title{A Sample {\ttlit Proceedings of the VLDB Endowment} Paper in LaTeX
%Format\titlenote{for use with vldb.cls}}

\title{Stream Frequency Over Interval Queries
\thanks{
This work is partially supported by ISF grant \#1505/16, Technion HPI research school, Zuckerman Institute and the \mbox{Technion Hiroshi Fujiwara cyber security research center.}}
%\thanks{\small
%This work is supported by ISF grant \#1505/16 and the HPI.}
}

\ifdefined\arXiv
%\author{Ran Ben Basat \qquad Roy Friedman \qquad Rana Shahout\\
%    Computer Science, Technion\qquad{}\texttt{\{sran,roy,ranas\}@cs.technion.ac.il}}
%\author{Ran Ben Basat \qquad Roy Friedman \qquad Rana Shahout\\
%    Computer Science, Technion\qquad{}\texttt{\{sran,roy,ranas\}@cs.technion.ac.il}}
\author[1]{Ran Ben Basat\thanks{ran@seas.harvard.edu}}
\author[2]{Roy Friedman\thanks{roy@cs.technion.ac.il}}
\author[2]{Rana Shahout\thanks{ranas@cs.technion.ac.il}}
\affil[1]{Computer Science, Harvard University}
\affil[2]{Computer Science, Technion}
\else

% possible, but not really needed or used for PVLDB:
%\subtitle{[Extended Abstract]
%\titlenote{A full version of this paper is available as\textit{Author's Guide to Preparing ACM SIG Proceedings Using \LaTeX$2_\epsilon$\ and BibTeX} at \texttt{www.acm.org/eaddress.htm}}}

% ****************** AUTHORS **************************************


\numberofauthors{3} %  in this sample file, there are a *total*

\author{
% 1st. author
\alignauthor
Ran Ben Basat\\
       \affaddr{Harvard University}\\
       \email{ran@seas.harvard.edu}
% 2nd. author
\alignauthor
Roy Friedman\\
       \affaddr{CS Technion}\\
       \email{roy@cs.technion.ac.il}
% 3rd. author
\alignauthor
Rana Shahout\\
       \affaddr{CS Technion}\\
	\email{ranas@cs.technion.ac.il}
}
\fi

\setlength{\textfloatsep}{14pt}
\setlength{\intextsep}{16pt}
\setlength{\dbltextfloatsep}{12pt}
\setlength{\abovedisplayskip}{3pt}
\setlength{\belowdisplayskip}{3pt}
\titlespacing*{\section} {0pt}{1.2ex plus .7ex minus .2ex}{1.2ex plus .2ex}
\titlespacing*{\subsection} {0pt}{1.1ex plus .7ex minus .2ex}{1.00ex plus .2ex}
\titlespacing*{\subsubsection}{0pt}{0.9ex plus .7ex minus .2ex}{0.90ex plus .2ex}

\maketitle
\begin{abstract}

Stream frequency measurements are fundamental in many data stream applications such as financial data trackers, intrusion-detection systems, and network monitoring.
Typically, recent data items are more relevant than old ones, a notion we can capture through a \emph{sliding window} abstraction.
This paper considers a generalized sliding window model that supports stream frequency queries over an interval given at \emph{query time}.
This enables drill-down queries, in which we can examine the behavior of the system in finer and finer granularities.
For this model, we asymptotically improve the space bounds of existing work, reduce the update and query time to a constant, and provide deterministic solutions.
When evaluated over real Internet packet traces, our fastest algorithm processes items $90$--$250$ times faster, serves queries at least $730$ times quicker and consumes at least $40\%$ less space than the best known method.

%While existing works provide probabilistic accuracy guarantees, we introduce efficient deterministic algorithms.
%Our methods process packets in constant time and asymptotically improve the naive approach of running multiple instances of fixed window size solutions.
%One algorithm is asymptotically space optimal while serving queries in logarithmic time, whereas the second algorithm answers queries in constant time and incurs a sub-quadratic space overhead.
\end{abstract}

%\end{document}\endinput
\maketitle
\section{Introduction}
\label{sec:intro}

%\subsection{Background}

High-performance stream processing is essential for many applications such as financial data trackers, intrusion-detection systems, network monitoring, and sensor networks.
Such applications require algorithms that are both time and space efficient to cope with high-speed data streams.
Space efficiency is needed, due to the memory hierarchy structure, to enable cache residency and to avoid page swapping. This residency is vital for obtaining good performance, even when the theoretical computational cost is small (e.g., constant time algorithms may be inefficient if they access the DRAM for each element).
To that end, stream processing algorithms often build compact approximate \emph{sketches} (synopses) of the input streams.
%As a trade-off, these sketches answers approximate query with guarantees on specific error.

Recent items are often more relevant than old ones, which requires an aging mechanism for the sketches.
Many applications realize this by tracking the stream's items over a \emph{sliding window}.
That is, the sliding window model~\cite{DatarGIM02} considers only a window of the most recent items in the stream,  while older ones do not affect the quantity we wish to estimate.
Indeed, the problem of maintaining different types of sliding window statistics \mbox{was extensively studied~\cite{ArasuM04,WCSS,DatarGIM02,papapetrou2015sketching,LeeT06}.}

Yet, sometimes the window of interest may not be known a priori \changed{or they may be multiple interesting windows~\cite{dallachiesa2013identifying}}.
Further, the ability to perform drill-down queries, in which we examine the behavior of the system in finer and finer granularity may also be beneficial, especially for security applications.
For example, this enables detecting when precisely a particular anomaly has started and who was involved in it~\changed{\cite{Estan2002}}.
Additional applications for this capability include identifying the sources of flash crowd effects and pinpointing the cause-effect relation surrounding a surge in demand on an e-commerce website~\changed{\cite{ecommerce}}.

In this work, we study a model that allows the user to specify an interval of interest at query time.
This extends traditional sliding windows that only consider fixed sized windows.
%where the measurement is done with regards to a fixed sized window.
As depicted in Figure~\ref{fig:overview}, a sub-interval of a maximal window is passed as a parameter for each query, and the goal of the algorithm is to reply correspondingly.
Naturally, one could maintain an instance of a sliding window algorithm for each possible interval within the maximal sliding window.
Alas, this is both computationally and space inefficient.
\mbox{Hence, the challenge is to devise efficient solutions.}

This same model was previously explored in~\cite{papapetrou2015sketching}, which based their solution on \emph{exponential histograms}~\cite{DatarGIM02}.
However, as we elaborate below, their solution is both memory wasteful and computationally inefficient.
Further, they only provide probabilistic guarantees.

\subsection*{Contributions}

\begin{table*}[t]
	\scriptsize
	\addtolength{\tabcolsep}{-0pt} \addtolength{\parskip}{-0.1mm} 
	\centering{
	%\resizebox{\textwidth}{!}{
		\begin{tabular}{|@{}c@{}|c|c|c|c|>{\tiny}c|}
			\hline
			Algorithm & Space &  Update Time & Query Time & Comments\tabularnewline
			\hline
			\hline
			WCSS ~\cite{WCSS}& $O(\eps^{-1}\log (W|\mathcal U|))$ & $O(1)$ & $O(1)$ & Only supports fixed-size window queries.
			\tabularnewline
			\hline
			ECM ~\cite{papapetrou2015sketching}& $O(\eps^{-2}\log W\log\delta^{-1})$ & $O(\log\delta^{-1})$ & $O(\eps^{-1}\log W\log\delta^{-1})$ & Only provides probabilistic guarantees.\tabularnewline
			\hline
			RAW & $O(\eps^{-2}\log (W|\mathcal U|))$ & $O(\oneOverE)$ & $O(1)$ & Uses prior art (WCSS) as a black box.\tabularnewline
			\hline
			\multirow{3}{*}{ACC$_k$} & \multirow{1}{*}{\hspace{-15mm}$O\Big(\oneOverE\log (W|\mathcal U|)$} & \multirow{3}{*}{$O(k+\eps^{-2}/W)$} & \multirow{3}{*}{$O(k)$} & Constant time operations for \tabularnewline
            & $+k\eps^{-(1+1/k)} \log\oneOverE\Big)$& & & $k=O(1)\wedge \eps=\Omega(W^{-1/2})$.\tabularnewline
			\hline
%\ifdefined\EXTENDED		
%			HIT & $O(\oneOverE(\log W + \log|\mathcal U|))$ & $O(1 + \log \oneOverE\cdot \oneOverE/W)$ & $O(\log \oneOverE)$ & A logarithmic factor from optimal space.\tabularnewline & & & & $O(1)$ time updates when $\eps=\Omega\parentheses{\frac{\logw}{W}}$.
%\fi
%\ifdefined\NINEPAGES
			\multirow{3}{*}{HIT} & \multirow{3}{*}{$O(\eps^{-1}(\log (W|\mathcal U|) + \log^2\oneOverE))$} & \multirow{3}{*}{\changed{$O(1 + \parentheses{\oneOverE\cdot\log{ \oneOverE} }/W)$}} & \multirow{3}{*}{$O(\log \oneOverE)$} & Optimal space when $\log^2\oneOverE=O(\log (W|\mathcal U|))$,
			\tabularnewline & & & &
			 $O(1)$ time updates when $\eps=\Omega\parentheses{\frac{\logw}{W}}$.
%\fi	
			\tabularnewline
			\hline
		\end{tabular}
		\normalsize	
		{{\caption{Comparison of the algorithms proposed in the paper with ECM and WCSS (that solves the simpler problem of fixed-size windows). ACC$_k$ can be instantiated for any $k\in\mathbb N$.
\ifdefined\NINEPAGES
\vspace*{-0.2cm}
\fi				
						}
				\label{tbl:comparison}
	}}}
\end{table*}
%\section{Contributions}

\begin{figure}[t]
	%\medskip
	\centering
	%\frame{
	\includegraphics[width=\linewidth]{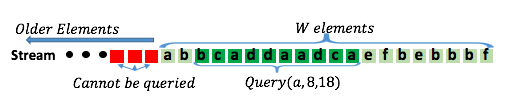}
	%}
	\vspace{-8mm}	\caption{\small{
    We process items and support frequency queries within an interval specified at query time.
    While the traditional sliding window model can answer queries for a \emph{fixed window}, our approach allows us to consider \emph{any} interval that is contained within the last $W$ items.
    In this example, we ask about the frequency of the item $a$ within the interval $[8,18]$. If we allow an additive error of $2$, the answer to this query should be in the range $[4, 6]$.}
	}
	%\caption{\small{We process items and support frequency queries within an interval specified at query time.
%	While previous work can answer queries for a \emph{fixed window}, our approach allows \emph{any} interval that is contained in the last $W$ packets to be queried.
	%While the traditional sliding window model can answer queries for a \emph{fixed window}, our approach allows us to consider \emph{any} interval that is contained within the last $W$ items.
	%In this example, assuming that we allow an additive error of $1$, we ask about the frequency of item $a$ within the interval $[8,18]$. The answer to this query should be in the range $[4, 5]$.}
	%}
	\label{fig:overview}
	\vspace{-3mm}
\end{figure}
\normalsize

Our work focuses on the problem of estimating frequencies over an ad-hoc interval given at query time. 
We start by introducing a formal definition of this generalized estimation problem nicknamed \IFProblem{}.
%and \IHHProblem.

To systematically explore the problem, we first present a na\"{\i}ve strawman algorithm (RAW), which uses multiple instances of a state-of-the-art fixed window algorithm.
%\ifdefined\EXTENDED
In such an approach, an interval query is satisfied by querying the instances that are closest to the beginning and end of the interval and then subtracting their results.
This algorithm is memory wasteful and its update time is slow, but it serves as a baseline for comparing our more sophisticated solutions.
%\fi
%\ifdefined\NINEPAGES
%It serves as a baseline for comparing our more sophisticated solutions.
%\fi
Interestingly, RAW achieves constant query time while the previously published ECM algorithm~\cite{papapetrou2015sketching} answers queries in $O(\eps^{-1}\log W\log\delta^{-1})$, where $W$ is the maximal window size and $\delta$ is the probability of failure. Additionally, it requires about the same amount of memory and is \mbox{deterministic while ECM has an error probability.}
%Moreover, RAW performs updates in $O(\oneOverE)$ and ECM in $O(\log\delta^{-1})$.

While developing our advanced algorithms, we discovered that both intrinsically solve a common problem that we nickname \SIProblem.
Hence, our next contribution is in identifying and formally defining the \SIProblem{} problem and showing a reduction from \SIProblem{} to the \IFProblem{} problem.
This makes our algorithms shorter, simpler, and easier to prove, \mbox{analyze and implement.}

Our algorithms, nicknamed \emph{HIT} and $ACC_k$ (to be precise, $\set{ACC_k}_{k\ge 1}$ is a family of algorithms), process items in constant time (under reasonable assumptions on the error target) -- asymptotically faster than RAW.
HIT is asymptotically memory optimal while serving queries in logarithmic time.
Conversely, $ACC_k$ answers queries in constant time and incurs a sub-quadratic space overhead.
%We formally prove the correctness and space analysis of our

We present formal correctness proofs as well as space and runtime analysis.
We summarize our solutions' asymptotic performance in Table~\ref{tbl:comparison}.

Our next contribution is a performance evaluation study of our various algorithms along with ($i$) ECM-Sketch ~\cite{papapetrou2015sketching}, the previously suggested solution for interval queries and ($ii$) the state-of-the-art \emph{fixed window} algorithm (WCSS)~\cite{WCSS}, which serves as a best case reference point since it solves a more straightforward problem.
We use on real-world packet traces from Internet backbone routers, from a university datacenter, and from a university's border router.
Overall, our methods (HIT and ACC$_k$) process items $75$--$2000$ times faster and consume at least $20$ times less space than the naive approach (RAW) while requiring a similar amount of memory as the state-of-the-art \emph{fixed size window} algorithm (WCSS).
Compared to the previously known solution to this problem (ECM-Sketch~\cite{papapetrou2015sketching}), all our advanced algorithms are both faster and more space efficient.
In particular, our fastest algorithm, ACC$_1$, processes items $90$--$250$ times faster than ECM-Sketch, serves queries at least $730$ times quicker and consumes at least $40\%$ less~space.

Last, we extend our results to time-based intervals, heavy hitters~\cite{MG,berinde2010space}, hierarchical heavy hitters~\cite{CormodeHHH,Hershberger2005}, and for detecting traffic volume heavy-hitters~\cite{DIM-SUM}, i.e., when counting each flow's total traffic rather than item count.
We also discuss applying our algorithms in a distributed settings, in which measurements are recorded independently by multiple sites (e.g., multiple routers), and the goal is to obtain a global network analysis.

\paragraph*{Paper roadmap}
We briefly survey related work in Section~\ref{sec:related}.
We state the formal model and problem statement in Section~\ref{sec:prelim}.
Our na\"{\i}ve algorithm RAW is described in Section~\ref{windows-sizes}.
We present the auxiliary \SIProblem{} problem, which both our advanced algorithms solve and has a simple reduction to the \IFProblem{} problem, in Section~\ref{sec:block-interval}.
The improved algorithms, \emph{HIT} and $ACC_k$, are then described in Section~\ref{sec:improved_algos}.
The performance evaluation of our algorithms and their comparison to ECM-Sketch and WCSS is detailed in Section~\ref{sec:eval}.
Section~\ref{sec:extensions} discusses extensions of our work.
%We present the extensions of our work in Section~\ref{sec:extensions}.
Finally, we conclude with a discussion in Section~\ref{sec:discussion}.

\section{Related Work}
\label{sec:related}
%For streams with item deletions, linear sketches such as Count Sketch~\cite{CountSketch} and Count Min Sketch~\cite{CMSketch} were proposed.
%For \emph{weighted} streams, where we try to estimate flow volumes, the logarithmic time of MG was recently improved to a constant~\cite{DIM-SUM}.

Count Sketch~\cite{CountSketch} and Count Min Sketch~\cite{CMSketch} are perhaps the two most widely used sketches for maintaining item's frequency estimation over a stream. 
The problem of estimating item frequencies over sliding windows was first studied in~\cite{ArasuM04}.
For estimating frequency within a $\weps$ additive error over a $W$ sized window, their algorithm requires $O(\oneOverE\log^2\oneOverE\log W)$ bits.
This was then reduced to the optimal $O(\oneOverE\log W)$ bits~\cite{LeeT06}.
In~\cite{HungLT10}, Hung and Ting improved the update time to $O(1)$ while being able to find all heavy hitters in the optimal $O(\oneOverE)$ time.
Finally, the WCSS algorithm presented in~\cite{WCSS} also estimates item frequencies in constant time.
While some of these works also considered a variant in which the window can expand and shrink when processing updates~\cite{ArasuM04,LeeT06}, its size was increased/decreased by one at each update, and cannot be specified at query~time.

The most relevant paper that solves the same problem as our work is~\cite{papapetrou2015sketching}, who was the first to explore heavy hitters interval queries.
They introduced a sketching technique with probabilistic accuracy guarantees called \emph{Exponential Count-Min sketch} (ECM-Sketch).
ECM-Sketch combines \emph{Count-Min sketch}' structure~\cite{CMSketch} with \emph{Exponential Histograms}~\cite{DatarGIM02}.
Count-Min sketch is composed of a set of $d$ hash functions, and a 2-dimensional array of counters of width $w$ and depth $d$. To add an item $x$ of value $v_x$, Count-Min sketch increases the counters located at $CM[ j, h_j (x)]$ by $v_x$, for $1\le j\le d$. Point query for an item q is done by getting the minimum value of the corresponding cells.

\changed{Exponential Histograms~\cite{DatarGIM02} allow tracking of metrics over a sliding window to within a multiplicative error. Specifically, they allow one to estimate the number of $1$'s in a sliding window of a binary stream. To that end, they utilize a sequence of \emph{buckets} such that each bucket stores the timestamp of the oldest $1$ in the bucket. When a new element arrives, a new bucket is created for it; to save space, the histogram may merge older buckets. While the amortized update complexity is $O(1)$, some arriving elements may trigger an $O(\log W)$-long cascade of bucket merges.
%For a $(1+\eps)$-multiplicative approximation, they use $O(\eps^{-1})$ 
}

ECM-Sketch replaces each Count-Min counter with an Exponential Histogram. Adding an item $x$ to the structure is analogous to the case of the regular Count-Min sketch.
For each of the histograms $CM[$j$, h_j(x)]$, where $1\le j\le d$, the item is registered with time/count of its arrival and all expired information is removed from the Exponential Histogram.
To query item $x$ in range $r$, each of the corresponding $d$ histograms $E($j$, h_j(x, r))$, where $1\le j\le d$, computes the given query range.
The estimate value for the frequency of $x$ is $\min\limits_{j=1,\dots, d} E(j, h_j(x), r)$.
\changed{ While the Exponential Histogram counters estimate the counts within a multiplicative error, their combination with the Count-Min sketch changes the error guarantee to additive.
}

\changed{
An alternative approach for these interval queries was proposed in~\cite{dallachiesa2013identifying}. Their solution uses hCount~\cite{jin2003dynamically}, a sketch algorithm which is essentially identical to the Count-Min sketch. Unlike the ECM-Sketch, which uses a matrix of Exponential Histograms, \cite{dallachiesa2013identifying} uses a sequence of $\log (W/b)$ \emph{buckets} each of which is associated with an hCount instance. The smallest bucket is of size $b$ while the size of the $i$'th bucket is $b\cdot 2^{i-1}$. When queried, \cite{dallachiesa2013identifying} finds the buckets closest to the interval and queries the hCount instances. The paper does not provide any formal accuracy guarantees but shows that it has reasonable accuracy in practice. It seems that the memory used is $O(\log(W/b)\cdot \eps^{-1}\log\delta^{-1}\log W)$ bits while the actual error has two components: (i) an error of up to $b+W/4$ in the \emph{time axis} (when the queried interval is not fully aligned with the buckets); and (ii) an error of up to $W\eps$, with probability $1-\delta$, due to the hCount \mbox{instance used for the queried buckets.}
}

In \emph{other} domains, ad-hoc window queries were proposed and investigated. That is, the algorithm assumes a predetermined upper bound on the window size $W$, and the user could specify the actual window size $w\le W$ at query time. This model was studied for quantiles~\cite{Lin2004} and summing~\cite{slidingRanker}.

The problem of identifying the frequent items in a data stream, known as \emph{heavy hitters}, dates back to the 80's~\cite{MG}. There, Misra and Gries (MG) proposed a space optimal algorithm for computing an $N\epsilon$ additive approximation for the frequency of elements in an $N$-sized stream. Their algorithm had a runtime of $O(\log \oneOverE)$, which was improved to a constant~\cite{Freq1,BatchDecrement}. Later, the Space Saving (SS) algorithm was proposed~\cite{SpaceSavings} and shown to be empirically superior to prior art~(see also~\cite{SpaceSavingIsTheBest2010,SpaceSavingIsTheBest2011}).
%\ifdefined\EXTENDED
Surprisingly, Agarwal et al. recently showed that MG and SS are isomorphic~\cite{Mergeable}, in the sense that from a $k$-counters MG data structure one can compute the estimate that a $k+1$ SS algorithm would produce.
%\fi

The problem of \emph{hierarchical heavy hitters}, which has important security and anomaly detection applications~\cite{HHHMitzenmacher}, was previously addressed with the SS algorithm~\cite{HHHMitzenmacher}.
To estimate the number of packets that originate from a specific network (rather than a single IP source), it maintains several separate SS instances, each dedicated to measuring different network sizes (e.g., networks with 2-bytes net ids are tracked separately than those with 3-bytes, etc.).
When a packet arrives, all possible prefixes are computed and each is fed into the relevant SS~instance.
%\ifdefined\EXTENDED
Recently, it was shown that randomization techniques can drive the update complexity down to a constant~\cite{MASCOTS,HHH-SIGCOMM}.
%\fi

\section{Preliminaries}
\label{sec:prelim}

Given a \emph{universe} $\mathcal U$, a stream $\mathcal S = x_1,x_2,\ldots\in \mathcal U^*$ is a sequence of universe elements.
We denote by $W\in\mathbb N$ the \emph{maximal window size}; that is, we consider algorithms that answer queries for an interval contained with the last $W$ elements window.
\changed{The actual value of $W$ is application dependent. For example, a network operator that wishes to monitor up to a minute of traffic of a major backbone link may need $W$ of tens of millions of packets~\cite{CAIDACH16}}.
Given an element $x\in\mathcal U$ and an integer $0\le w\le W$, the \emph{$w$-frequency}, denoted $\xWindowFrequency$,  is the number of times $x$ appears within the last $w$ elements of $\mathcal S$.
For integers $i\le j\le W$, we further denote by $\xIntervalFrequency\triangleq \xWindowFrequency[j] - \xWindowFrequency[i]$ the frequency of $x$ between the $i^{th}$ and $j^{th}$ \mbox{most recent elements of $\mathcal S$.}

%Finally, for integers $i\le j\le W$ and a real number $\theta\in[0,1]$, we denote by $\IntervalHH\triangleq\IHHDef$ the set of heavy hitters items that appeared at least a $\theta$ fraction of the queried interval.
% in the last $\mathcal W$ elements of $\mathcal S$.

We seek algorithms that support the following operations:
\begin{itemize}
	\item {\sc \textbf{ADD}$\bm{(x)}$}: given an element $x\in\mathcal U$, append $x$ to $\mathcal S$.
	\item \BIFQ{}$\bm{(x,i,j)}$: given $x\in\mathcal U$ and indices $i\le j\le W$, return an estimate $\xIntervalFrequencyEstimator$ of $\xIntervalFrequency$.
%	\item \BIHHQ{}$\bm{(\theta,i,j)}$: given indices $i\le j\le W$, return an estimate $\IntervalHHEstimator\subseteq\mathcal U$ that approximates $\IntervalHH$.
%	\item {\sc \textbf{InterWinQuery$\mathbf{(x, i, j)}$}}: return an estimate $\widehat{{f_x}^{W_{i, j}}}$ of ${f_x}^{W_{i, j}}$.
\end{itemize}

We now formalize the required guarantees.
\begin{definition}
An algorithm solves \IFProblem{} if given \mbox{any \IFQ{}$(x,i,j)$ it satisfies} \IFQGuarantee
\end{definition}
%\textbf{Problem definition:}
%\begin{itemize}
%%	\item $\mathbf{(W,\epsilon)}$\textbf{-}{\sc \textbf{Window Estimation}}: %\winQuery{} returns an estimation ($\widehat{f^W_x}$) that satisfies
%%	$$f^W_x \le \widehat{f^W_x} \le f^W_x + W\cdot \epsilon.$$
%%	\normalsize
%	\item \IFProblem: \\
%	\IFQ{}$(x,i,j)$ satisfies
%	\IFQGuarantee
%%	\item \IHHProblem: \\
%%	\IHHQ{}$(\theta,i,j)$ satisfies
%%	$$\IntervalHH\subseteq \IntervalHHEstimator \subseteq \set{x\in\mathcal U\mid %\xIntervalFrequency\ge \intervalHHThreshold - \epsError} .$$	
%%	That is, the estimated set must contain all elements that appear at least a %$\theta$ fraction of the interval and should not have any members whose frequency %is lower than $\theta\cdotpa{j-i} - \epsError$.
%	\normalsize
%\end{itemize}
For simplicity of presentation, we assume that $\epsError/12$ and $\oneOverE$ are integers.
For ease of reference, Table~\ref{tbl:notations} includes a summary of basic notations used in this work.

\begin{table}[t]
	\centering
	\scriptsize
	\begin{tabular}{|c|l|}
		\hline
		Symbol & Meaning \tabularnewline
		\hline
		$\mathcal S$ & the data stream \tabularnewline
		\hline
		$\mathcal U$ & the universe of elements \tabularnewline
		\hline
%		$N$ & number of elements in the stream \tabularnewline
%		\hline
		$W$ & the maximal window size \tabularnewline
		\hline
%		$f_x$ & the frequency of an element $x$ in $\mathcal  S$ \tabularnewline
%		\hline
%		$\widehat{f_x}$ & an estimation of $f_x$ \tabularnewline
%		\hline
		\xWindowFrequency & \parbox[t]{2.5in}{ the frequency of element $x$ within the last $w$ elements of $\mathcal S$\strut} \tabularnewline
		\hline
		\xWindowFrequencyEstimator & an estimation of \xWindowFrequency \tabularnewline
		\hline
		\xIntervalFrequency & \parbox[t]{2.5in}{ the frequency of element $x$ between the $i^{th}$ and $j^{th}$ most recent elements of $\mathcal S$\strut} \tabularnewline
		\hline
		\xIntervalFrequencyEstimator & an estimation of \xIntervalFrequency \tabularnewline
		\hline		
		$\epsilon$ & estimation accuracy parameter \tabularnewline
		\hline
		$\delta$ & probability of failure \tabularnewline
		\hline
%		$\theta$ & heavy hitters threshold \tabularnewline
%		\hline
        n & number of blocks in a frame ($6/\epsilon$) \tabularnewline
        \hline
        N & max sum of blocks' cardinalities ($12/\epsilon$) within a window \tabularnewline
        \hline
\ifdefined\EXTENDED				
		\IntervalHH & the interval's heavy hitters -- \IHHDef\tabularnewline
		\hline		
		\IntervalHHEstimator & estimation of the heavy hitters set \tabularnewline
		\hline		
\fi		
	\end{tabular}
	\caption{List of Symbols}
	\label{tbl:notations}
	\vspace{-2em}
\end{table}
\normalsize

\textbf{Space Saving:} as we use the Space Saving (SS) algorithm~\cite{SpaceSavings} in our reduction in Section~\ref{sec:reduction}, we overview it here. SS maintains a set of $1/\epsilon$ counters, each has an associated element and a value. When an item arrives, SS first checks if it has a counter. If so, the counter is incremented; otherwise, SS allocates the item with a minimal-valued counter. For example, assume that the smallest counter was associated with $x$ and had a value of $4$; if $y$ arrives and has no counter, it will take over $x$'s counter and increment its value to $5$ (leaving $x$ without a counter). When queried for the frequency of a flow, we return the value of its counter if it has one, or the minimal counter's value otherwise. If we denote the overall number of insertions by $Z$, then we have that the sum of counters equals $Z$, and the minimal counter is at most $Z\epsilon$. This ensures that the error in the SS estimate is at most $Z\epsilon$. An important observation is that once a counter reached a value of $Z\epsilon$ it is no longer the minimum throughout the rest of the measurement.
%\input{HH2Frequency}
%\subsection{Redundant Approximate Window (RAW)}
\section{Strawman Algorithm}
\label{windows-sizes}
\ifdefined\NINEPAGES
%\begin{figure}[]
\begin{figure*}[t]
\else
\begin{figure*}[t]
\fi
	\centering
	%\frame{
\ifdefined\NINEPAGES
	\includegraphics[width=0.8\linewidth]{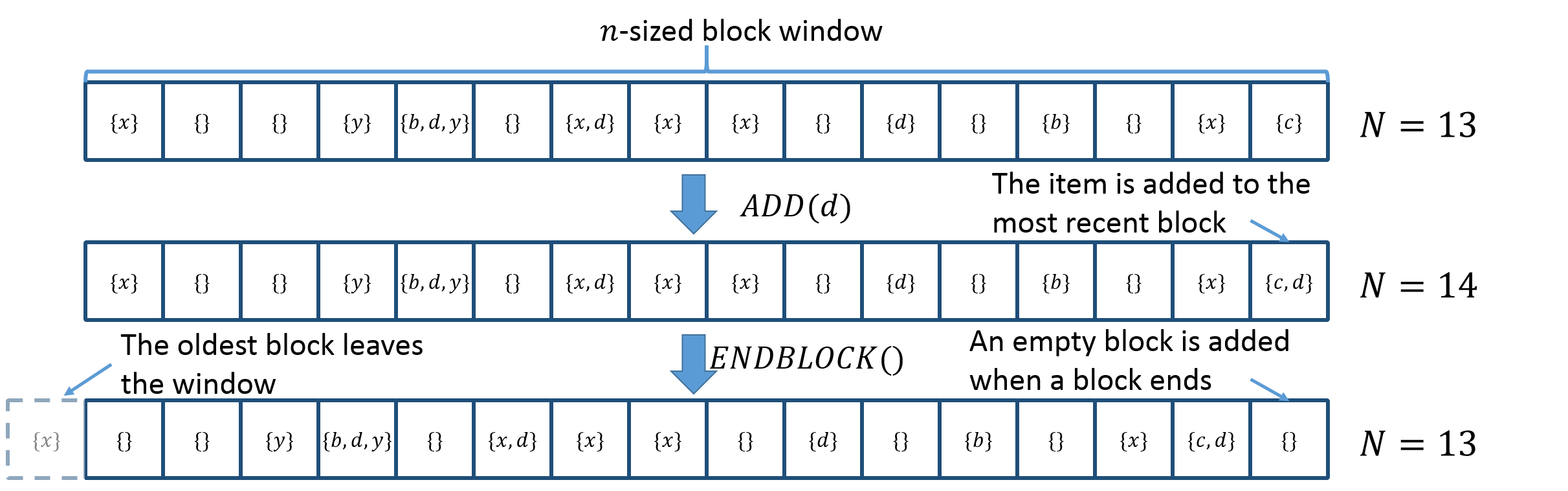}
\else
	\includegraphics[width=0.8\linewidth]{setStreamIllustration.png}
\fi
	%}
	\caption{The block stream setting. Here, after the {\sc EndBlock}, $x$ appears in two blocks out of the last $9$ and thus $\xSetWindowFrequency[9]=2$.}
	\label{fig:SetStream}
	\vspace{-0.3em}
	\ifdefined\NINEPAGES	
%	\vspace{-0.3em}
	\fi	
\ifdefined\NINEPAGES
%\end{figure}
\end{figure*}
\else
\end{figure*}
\fi
Here, we present the simple Redundant Approximate Windows (RAW) algorithm that uses several instances of a black box algorithm $\mathbb A(w,\eps)$ for solving the frequency estimation problem over a fixed $W$-sized window.
That is, we assume that $\mathbb A(w,\eps)$ supports the {\sc ADD$(x)$} operation and upon {\sc Query$(x)$} produces an estimation $\xWindowFrequencyEstimator$ that satisfies:
\ifdefined\EXTENDED
$$\xWindowFrequency \le \xWindowFrequencyEstimator \le \xWindowFrequency + \epsError[w].$$
\else
$\xWindowFrequency \le \xWindowFrequencyEstimator \le \xWindowFrequency + \epsError[w].$
\fi
We note that the WCSS algorithm~\cite{WCSS} solves this problem using $O(\oneOverE)$ counters and in $O(1)$ time for updates and queries. Both its runtime and space are optimal.\footnote{A lower bound of matching asymptotic complexity appears in~\cite{BatchDecrement}, even for non-window solutions.}

Specifically, we maintain $4\oneOverE$ separate solutions denoted $A_1,\ldots A_{4\oneOverE}$, where each $A_\ell$ is an $\mathbb A(\ell\cdot\epsError/4,\eps/4)$ instance.
We perform the {\sc ADD}$(x)$ operation simply by invoking the operation $A_\ell.${\sc ADD}$(x)$ for $\ell=1,\ldots,4\epsilon^{-1}$.
When given an \IFQ{}$(x,i,j)$, we return
%\ifdefined\EXTENDED
\ifdefined\arXiv
\begin{align}
\xIntervalFrequencyEstimator \triangleq A_{\ceil{j/(\epsError/4)}}.\mbox{{\sc Query$(x)$}} - A_{\floor{i/(\epsError/4)}}.\mbox{{\sc Query$(x)$}} +\epsError/4.\label{RAWEstimator}
\end{align}
\else
\begin{multline*}
\xIntervalFrequencyEstimator \triangleq A_{\ceil{j/(\epsError/4)}}.\mbox{{\sc Query$(x)$}} \\- A_{\floor{i/(\epsError/4)}}.\mbox{{\sc Query$(x)$}} +\epsError/4.%\label{RAWEstimator}
\end{multline*}
\fi
%\else
%$\xIntervalFrequencyEstimator \triangleq A_{\ceil{j/(\epsError/4)}}.\mbox{{\sc Query$(x)$}} - A_{\floor{i/(\epsError/4)}}.\mbox{{\sc Query$(x)$}} +\epsError/4.$
%\fi
We now state the correctness of RAW. Due to lack of space, we defer the proof to the full version of the paper~\cite{full-version}.
\ifdefined\NINEPAGES
\begin{theorem}
	Using WCSS as the black box algorithm $\mathbb A$, RAW requires $O(\eps^{-2}(\log W + \log|\mathcal U|))$ bits, performs updates in $O(\oneOverE)$ time and answers queries in constant time.
\end{theorem}
\else
Next, we analyze the properties of RAW.
\begin{theorem}
Let $\mathbb A$ be a black box algorithm as above that uses $S(w,\eps)$ space and runs at $U(w,\eps)$ time for updates and $Q(w,\eps)$ time for queries.
Then RAW requires $O(\oneOverE S(w,\eps))$ space, performs updates in $O(\oneOverE U(w,\eps))$ time, and answers queries in $O(Q(w,\eps))$ time.
Further, RAW solves the \IFProblem{} problem.
\end{theorem}
\begin{proof}
The run times above follows immediately from the fact that RAW utilizes $O(\oneOverE)$ instances of $\mathbb A(\cdot, \eps/4)$, updates each of them when processing elements, and queries only two instances per interval query.
Next, we will prove the correctness of RAW.
%Assume that the index of the most recent item is $N$, i.e.,
%\begin{align*}
%\xIntervalFrequency = \Big|\big\{\ell\in\set{N-i+1,\ldots,N-j}\mid x_\ell = x\big\}\Big|.
%\end{align*}

Notice that we can express the interval frequency as:
\ifdefined\arXiv
\begin{align}
\xIntervalFrequency &= \xWindowFrequency[j] - \xWindowFrequency[i]
= \xWindowFrequency[\ceil{j/(\epsError/4)}\cdot\epsError/4] - \xWindowFrequency[\floor{i/(\epsError/4)}\cdot\epsError/4]
 - \xParameterizedIntervalFrequency{j}{\ceil{j/(\epsError/4)}\cdot\epsError/4}
 - \xParameterizedIntervalFrequency{\floor{i/(\epsError/4)}\cdot\epsError/4}{i}.\label{eq1}
\end{align}
\else
\begin{multline}
\xIntervalFrequency = \xWindowFrequency[j] - \xWindowFrequency[i]\\
= \xWindowFrequency[\ceil{j/(\epsError/4)}\cdot\epsError/4] - \xWindowFrequency[\floor{i/(\epsError/4)}\cdot\epsError/4]\\
 - \xParameterizedIntervalFrequency{\ceil{j/(\epsError/4)}\cdot\epsError/4}{i}
 - \xParameterizedIntervalFrequency{i}{\floor{i/(\epsError/4)}\cdot\epsError/4}
.\label{eq1}
\end{multline}
\fi
Next, we note that
$$\forall n,d\in\mathbb N: 0\le|n-d\floor{n/d}|,|d\ceil{n/d}-n|\le d,$$
and since $\forall a\ge b: 0\le \xParameterizedIntervalFrequency{a}{b} \le b-a$, we have
$$0 \le \xParameterizedIntervalFrequency{j}{\ceil{j/(\epsError/4)}\cdot\epsError/4}, \xParameterizedIntervalFrequency{\floor{i/(\epsError/4)}\cdot\epsError/4}{i} \le \epsError/4.$$ 
Plugging this into \eqref{eq1}, we get
\begin{align*}
\begin{cases}
\xIntervalFrequency \le \xWindowFrequency[\ceil{j/(\epsError/4)}\cdot\epsError/4] - \xWindowFrequency[\floor{i/(\epsError/4)}\cdot\epsError/4]\\
\xIntervalFrequency \ge \xWindowFrequency[\ceil{j/(\epsError/4)}\cdot\epsError/4] - \xWindowFrequency[\floor{i/(\epsError/4)}\cdot\epsError/4] - \epsError/2.
\end{cases}
\end{align*}
Now our estimation in \eqref{RAWEstimator} relies on the estimations produced by $A_{\ceil{i/(\epsError/4)}},A_{\floor{j/(\epsError/4)}}$. By the correctness of $\mathbb A$, we are guaranteed that
\begin{align}
\begin{cases}
A_{\ceil{j/(\epsError/4)}}.\mbox{{\sc Query$(x)$}}\ge \xWindowFrequency[\ceil{j/(\epsError/4)}\cdot\epsError/4]  \\
A_{\ceil{j/(\epsError/4)}}.\mbox{{\sc Query$(x)$}} \le \xWindowFrequency[\ceil{j/(\epsError/4)}\cdot\epsError/4] + \epsError/4\\
A_{\floor{i/(\epsError/4)}}.\mbox{{\sc Query$(x)$}}\ge \xWindowFrequency[\floor{i/(\epsError/4)}\cdot\epsError/4]  \\
A_{\floor{i/(\epsError/4)}}.\mbox{{\sc Query$(x)$}} \le \xWindowFrequency[\floor{i/(\epsError/4)}\cdot\epsError/4] + \epsError/4.
\end{cases}\label{eq2}
\end{align}
Combining \eqref{eq2} with \eqref{RAWEstimator} we establish
\ifdefined\arXiv
\begin{multline}
\xIntervalFrequencyEstimator \ge (\xWindowFrequency[\ceil{j/(\epsError/4)}\cdot\epsError/4]) - (\xWindowFrequency[\floor{i/(\epsError/4)}\cdot\epsError/4]+\epsError/4) +\epsError/4
= \xWindowFrequency[\ceil{j/(\epsError/4)}\cdot\epsError/4] - \xWindowFrequency[\floor{i/(\epsError/4)}\cdot\epsError/4].\label{eq4}
\end{multline}
\else
\begin{multline}
\xIntervalFrequencyEstimator \ge (\xWindowFrequency[\ceil{j/(\epsError/4)}\cdot\epsError/4]) \\- (\xWindowFrequency[\floor{i/(\epsError/4)}\cdot\epsError/4]+\epsError/4) +\epsError/4
\\= \xWindowFrequency[\ceil{j/(\epsError/4)}\cdot\epsError/4] - \xWindowFrequency[\floor{i/(\epsError/4)}\cdot\epsError/4].\label{eq4}
\end{multline}
\fi
Similarly,
\ifdefined\arXiv
\begin{multline}
 \xIntervalFrequencyEstimator \le (\xWindowFrequency[\ceil{j/(\epsError/4)}\cdot\epsError/4] + \epsError/4) - (\xWindowFrequency[\floor{i/(\epsError/4)}\cdot\epsError/4]) +\epsError/4,
 = \xWindowFrequency[\ceil{j/(\epsError/4)}\cdot\epsError/4] - \xWindowFrequency[\floor{i/(\epsError/4)}\cdot\epsError/4]  + \epsError/2.\label{eq5}
\end{multline}
\else
\begin{multline}
 \xIntervalFrequencyEstimator \le (\xWindowFrequency[\ceil{j/(\epsError/4)}\cdot\epsError/4] + \epsError/4) \\- (\xWindowFrequency[\floor{i/(\epsError/4)}\cdot\epsError/4]) +\epsError/4,
 \\= \xWindowFrequency[\ceil{j/(\epsError/4)}\cdot\epsError/4] - \xWindowFrequency[\floor{i/(\epsError/4)}\cdot\epsError/4]  + \epsError/2.\label{eq5}
\end{multline}
\fi
Finally, we substitute \eqref{eq4} and \eqref{eq5} in \eqref{eq1} to obtain the desired \IFQGuarantee[\qedhere]
\end{proof}

While RAW does not assume anything about $\mathbb A$, WCSS was shown to be asymptotically optimal both in terms of runtime and memory~\cite{WCSS}.
Thus, obtaining an improved fixed-window algorithm can only allow constant factor reductions in time and space.
Also, while $\mathbb A$'s error is proportional to the window size (i.e., the error in the estimation of $A_\ell$ is at most $\ell\cdot\epsError/4$, which may be smaller than the $\epsError/4$ we used in the analysis), optimizing the error for each individual instance does not reduce the space by more than 50\%.
In the next section, we propose novel techniques to asymptotically reduce both space and update time.
Taking into account that every counter consists of an $O(\log|\mathcal U|)$ bits identifier and an $O(\log W)$ bits value, we conclude the following:
\begin{corollary}
Using WCSS as the black box algorithm $\mathbb A$, RAW requires $O(\eps^{-2}(\log W + \log|\mathcal U|))$ bits, performs updates in $O(\oneOverE)$ time and answers queries in constant time.
\end{corollary}
\fi

\section{Block Interval Frequency}
\label{sec:block-interval}
In this section, we formally define an auxiliary problem, nicknamed \SIProblem{},
show a reduction to the ${\problemParams}${-}
{IntervalFrequency} problem, and rigorously analyze the reduction's cost.
Our motivation lies in the fact that the suggested algorithms in Section~\ref{sec:improved_algos} both intrinsically solve the \SIProblem{} auxiliary problem. 
It also has the benefit that any improved reduction between these problems would improve both algorithms.
In \SIProblem{}, the arriving elements are inserted into $O(W\epsilon)\text{-sized}$ ``blocks'' and we are required to compute exact interval frequencies \emph{within the blocks}.
Doing so simplifies the presentation and analysis of the algorithms in Section~\ref{sec:improved_algos}, in which we propose algorithms that improve over RAW in both space and update time.
The two algorithms, $HIT$ and $ACC_k$ present a space-time tradeoff while achieving asymptotic reductions over RAW.

\subsection{The Block Interval Frequency Problem}
\label{block_interval_definition}
Here, instead of frequency, we consider items' \emph{block frequency}.
Namely, for some $x\in\mathcal U$, we define its window block frequency $\xSetWindowFrequency$ as the number of blocks $x$ appears in within the last $n$ blocks.
For integers $i\le j\le n$, we define $\xSetIntervalFrequency\triangleq \xSetWindowFrequency[j] - \xSetWindowFrequency[i]$.
Block algorithms support three operations:
\begin{itemize}
	\item {\sc \textbf{ADD}$\bm{(x)}$}: given an element $x\in\mathcal U$, add it to~the~stream.
	\item {\sc \textbf{ENDBLOCK}$\bm{()}$}: a new empty block is inserted into the window, and the oldest one leaves.
	\item \BSIQ{}$\bm{(x,i,j)}$: given an element $x\in\mathcal U$ and indices $i\le j\le n$, compute $\xSetIntervalFrequency$ (without error).
\end{itemize}
We define the \SIProblem{} as ${(W,\eps=0)}${-}{\IFText}. That is, we say that an algorithm solves the \SIProblem{} problem if given an \SIQ{}${(x,i,j)}$ it is able to compute the \emph{exact} answer for any $i\le j\le n$ and $x\in\mathcal U$.

For analyzing the memory requirements of algorithms solving this problem, we denote by $N$ the sum of cardinalities of the blocks in the $n$-sized window.
An example of this {setting is given in Figure~\ref{fig:SetStream}.}

%\subsection{From Exact Block Frequencies to~Approximate~Interval Frequencies}
\subsection{A Reduction to \IFProblem{}}
\label{sec:reduction}
We 
%now
show a reduction from the \SIProblem{} problem to \IFProblem{}.
To that end, we assume that $\mathcal A$ is an algorithm that solves the $n$-\SIText{} problem for~$n\triangleq 6\oneOverE$.

Our reduction relies on the observation that by applying such $\mathcal A$ on a data structure maintained by counter-based algorithm such as Space Saving~\cite{SpaceSavings}, we can compute interval queries and not only fixed window size frequency estimations.
%As in WCSS~\cite{WCSS},
%\ifdefined\NINEPAGES
%\else
The setup of the reduction is illustrated in Figure~\ref{fig:window-counting}.
\begin{figure}[]
	%\medskip
	\centering
	%\frame{
	\ifdefined\arXiv
	\includegraphics[width=0.5\linewidth]{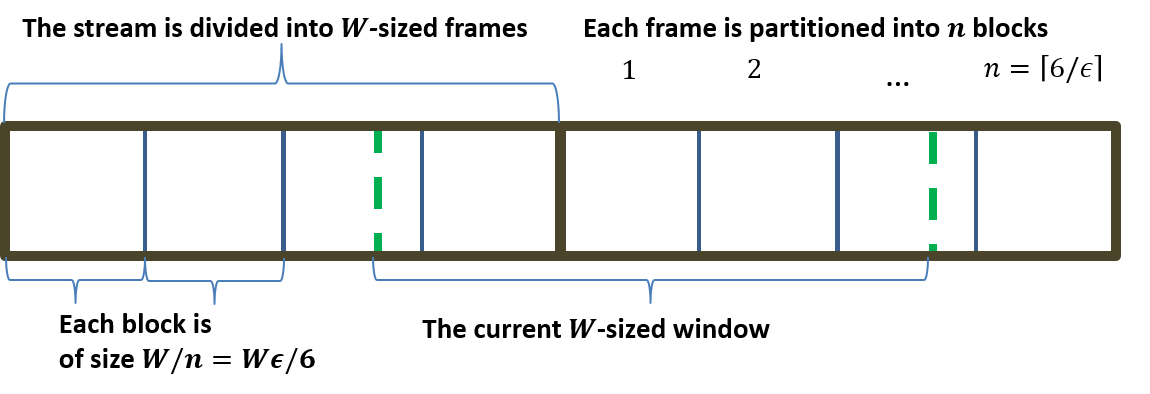}
	\else
	\includegraphics[width=0.9\linewidth]{sliding-window-setup.png}
	\fi
	%}
	\caption{The stream is logically divided into intervals of size $\window$ called \emph{frames} and each frame is logically partitioned into $\numBlocks$ equal-sized \emph{blocks}. The window of interest is also of size $\window$, and overlaps with at most $2$ frames and $\numBlocks+1$ blocks.
	\ifdefined\VLDB
	\vspace{-4mm}
	\fi
	}
	\label{fig:window-counting}
\end{figure}
%\fi
We break the stream into $W$ sized \emph{frames}, which are further divided into \emph{blocks} of size $O(W\eps)$.
We employ a Space Saving~\cite{SpaceSavings,WCSS} instance to track element frequencies within each frame; \changed{it supports two methods: \mbox{{\sc Add$(x)$}} -- adds element $x$ to the stream and \mbox{{\sc Query$(x)$}} -- reports the frequency estimation of element $x$ with tight guarantees on the error}.

Whenever a counter reaches an integer multiple of the block size, we add its associated flow's identifier to the most recent block of $\mathcal A$.
%That is, each set of the set stream represents a block in the input element sequence.
When a frame ends, we \emph{flush} the Space Saving instance and reset all of its counters.
We note that an implementation that supports constant time flush operations was suggested in~\cite{WCSS}.
Also, the max sum of block's cardinalities within a window (overlapping up to 2 frames) is $N=12/\epsilon$.
Finally, we reduce each \IFQ{} to an \SIQ{} by computing the indices of the blocks in which the interval starts and ends.
The variables of the reduction algorithm are described in Table~\ref{tbl:reductionVars} and its pseudocode appears in Algorithm~\ref{alg:reduction}.

\begin{table}[]

	\scriptsize
	\centering
	\begin{tabular}{|c|p{6.6cm}|}
		\hline
		% & Type \tabularnewline
		%\hline
		\frameOffset & the offset within the current frame.
		\tabularnewline
		\hline
		$\mathcal A$ & an algorithm that solves $(6/\eps)$-\SIText{}.
		\tabularnewline
		\hline
		\SSInstance{} & a Space Saving instance with $\ceil{6\oneOverE}$ counters.
		\tabularnewline
		\hline
		\blockSize{} & the size of blocks (fixed at $\blockSize\triangleq W\eps/6$).
		\tabularnewline
		\hline		
%		\OFFSET & The offset within the current frame. \tabularnewline
%		\hline
%		$\idtoidx$ & A hash map between item's ID and an index. \tabularnewline
%		\hline
	\end{tabular}
	\normalsize
	\caption{Variables used by the Algorithm~\ref{alg:reduction}.}
	\label{tbl:reductionVars}
	\ifdefined\VLDB
	\vspace{-1.4em}
	\fi
\end{table}

\ifdefined\NINEPAGES
\vspace*{-3mm}
\fi
\ifdefined\NINEPAGES
\setlength{\textfloatsep}{4pt}
\fi
\begin{algorithm}[H]

    \caption{From Blocks to Approximate Frequencies}\label{alg:reduction}
	\scriptsize
%\small
	%\begin{multicols}{2}
	\begin{algorithmic}[1]
		%		\ifdefined \NINEPAGES
		\Statex		Initialization: $%\BS \gets n^{1/k},
		\frameOffset\gets 0, \blockSize\triangleq W\eps/6,\quad \mbox{\textit{initialize}}\quad \mathcal A, \SSInstance(\eps/6).$
		\Function {Add}{$x$}
			\State $\frameOffset\gets (\frameOffset + 1)\mod W$
			\State $\SSInstance.Add(x)$\label{line:ssAdd}
			\If {$\SSInstance.Query(x) \mod \blockSize = 0$}\label{line:SSoverflow}
				\State $\mathcal A.Add(x)$\label{line:addToSet}
			\EndIf
			\If {$\frameOffset\mod \blockSize = 0$}
				\State $\mathcal A.EndBlock()$
			\EndIf
			\If {$\frameOffset = 0$}
				\State $\SSInstance.Flush()$
			\EndIf			
		\EndFunction	
		\Function {\IFQ{}}{$x,i,j$}
			\State \Return $\blockSize\cdotpa{\mathcal A.\SIQ(x,\ceil{i/\blockSize},\floor{j/\blockSize})+2}$\label{line:reductionQuery}
		\EndFunction
	\end{algorithmic}
	%\end{multicols}
	\normalsize
\end{algorithm}
\ifdefined\NINEPAGES
\setlength{\textfloatsep}{14pt}
\fi

\subsection{Theoretical Analysis}

%For proving the correctness of the reduction algorithm, we introduce some notations. We assume that the previous frame ended with element $\window$ and that the last item had an index of $\window+\frameOffset$.

Given a query \IFQ{}$(x,i,j)$, we are required to estimate $\xIntervalFrequency= \xWindowFrequency[j] - \xWindowFrequency[i]$.
Our estimator is $\xIntervalFrequencyEstimator= \mathcal A.\SIQ(x,\ceil{i/(W\eps/6)},\floor{j/(W\eps/6)})+W\eps/3$.
Intuitively, we query $\mathcal A$ for the block frequency of $x$ in the minimal sequence of blocks that contain interval $i,j$.
Every time $x$'s counter reaches an integer multiple of the block size, the condition in Line~\ref{line:SSoverflow} is satisfied and the block frequency of $x$, as tracked by $\mathcal A$, increases by $1$.
Thus, multiplying the block frequency by $s\triangleq W\eps/6$ allows us to approximate $x$'s frequency in the original stream.

There are several sources of estimation error:
First, we do not have a counter for each element but rather a Space Saving instance in which counters are shared.
Next, unless the counter of an item reaches an integer multiple of $\blockSize$, we do not add it to the block stream.
Additionally, the queried interval might not be aligned with the blocks.
Finally, when a frame ends, we flush the counters and thus lose the frequency counts of elements that are not recorded in the block stream.
With these sources of error in mind, we prove the correctness of our algorithm.
\begin{theorem}
\vspace{-1mm}
\label{theorem:block_size}
	Let $\mathcal A$ be an algorithm for the $6\oneOverE$-\SIText{} problem. Then Algorithm~\ref{alg:reduction} solves \IFProblem{}.
\end{theorem}
\vspace{-3mm}
\begin{proof}
We begin by noticing that once an element's counter reaches $\blockSize=W\eps/6$, it will stay associated with the element until the end of the frame.
This follows directly from the Space Saving algorithm, which only disassociates elements whose counter is minimal among all counters(see the SS overview in Section~\ref{sec:prelim}).
Recall that the number of elements in a frame is $W$ and that the Space Saving instance is allocated with $\ceil{6\oneOverE}$ counters.
Since the sum of counters always equals the number of elements processed, any counter that reaches a value of $\blockSize$ will never be minimal.
Thus, once an element was added to a block (Line~\ref{line:addToSet}), its block frequency within the frame is increased by one for every $\blockSize$ subsequent arrivals.
This means that an item might be added to a block while appearing just once in the stream, but this gives an overestimation of at most $\blockSize-1$.
As the queried intervals can overlap with two frames, this can happen at most twice, which imposes an overestimation error of no more than $2\blockSize$.

Our next error source is the fact that the queried interval may begin and end anywhere within a block.
By considering the blocks that contain $i$ and $j$, regardless of their offset, we incur another overestimation error of at most $2\blockSize$.

We have two sources of underestimation error, where items frequency is lower than $\blockSize$ times its block frequency.
The first is the count we lose when flushing the Space Saving instance.
Since we record every multiple of $\blockSize$ in the block stream, a frequency of at most $s-1$ is lost due to the flush.
Second, in the current frame, the residual frequency of an item (i.e., the appearances that have not been recorded in the block stream) may be at most $\blockSize-1$.
We make up for these by adding $2\blockSize$ to the estimation (Line~\ref{line:reductionQuery}).
As we have covered all error sources, the total error is smaller than $6\blockSize\le~W\eps$.
\end{proof}

\vspace{-4mm}\paragraph*{Reducing the Error}
Above, we used a block size of $\blockSize=W\eps/6$, which can be reduced to $W\eps/5$ as follows:
%The proof of Theorem~\ref{theorem:block_size} counts error sources.
One of the error sources in Theorem~\ref{theorem:block_size} is the fact that the queried interval $i,j$ may begin and end in the middle of a block and we always consider the entire blocks that contain $i$ and $j$.
We can optimize this by considering $i$'s and $j$'s offsets \emph{within} the relevant blocks, and including the block's frequency only if the offset crosses half the size of the block.
This incurs an overestimation error of at most $\blockSize$ instead of $2\blockSize$, allows blocks of size $W\eps/5$ and reduces the number of blocks to $n=5/\epsilon$.

\section{Improved Algorithms}
\label{sec:improved_algos}
%In general, in our algorithms we logically partition the stream into consecutive sequences of size $W$ called \emph{frames}.
%Each frame is further divided into $\numBlocks\triangleq\ceil{\frac{6}{\eps}}$ logical \emph{blocks}, each of size $\blockSize$, which we assume is an integer for simplicity.
%\textbf{Ran: This paragraph seems obsolete and out of context.}
%

\subsection{Approximate Cumulative Count ($\mathit{ACC}$)}
\label{sec:acc}
We present a family of algorithms for solving the \SIProblem{} problem.
Approximate Cumulative Count ($\mathit{ACC}$) of level $k$, denoted $\mathit{ACC_k}$, aims to compute the interval frequencies while accessing at most $k$ hash tables for updates and $2k+1$ for queries.
To reduce clutter, we assume in this section that $n^{1/k}\in\mathbb N$; this assumption can be omitted with the necessary adjustments while incurring a $1+o(1)$ multiplicative space overhead.
This family presents a space-time trade off --- the larger $k$ is, $\mathit{ACC_k}$ takes less space but is also slower.

The $\mathit{ACC}$ algorithms break the block stream into consecutive \emph{frames} of size $n$ (the maximal window size).
That is, blocks $B_1,B_2,\ldots B_n$ are in the first frame, $B_{n+1},\ldots B_{2n}$ in the second frame and so on.
Notice that any $n$-sized window intersects with at most two frames.
Within each frame, $\mathit{ACC}$ algorithms use a hierarchical structure of tables that enables it to compute an item's \mbox{block frequency in 
$O(1)$~time.}

$\mathit{ACC_1}$ and $\mathit{ACC_2}$ are illustrated in Figure~\ref{fig:ACC} and are explained below.
The simplest and fastest algorithm, $\mathit{ACC_1}$, computes for each block a frequency table that tracks how many times each item has arrived \emph{from the beginning of the frame}.
For example, the table for block $5n + 7$ (for $n>7$) will contain an entry for each item that is a member of at least one of $B_{5n+1},\ldots B_{5n+7}$.
The key is the item identifier, and the value is its block frequency from the frame's start.
This way, we can compute any block interval frequency by querying at most $3$ tables.
Within a frame, we can compute any interval by subtracting the queried item's block frequency at the beginning of the interval from its block frequency at the end. If the interval spans across two frames, we make \changed{one additional} query for reaching the beginning of the frame, in total \mbox{we query at most $3$ tables.}
\ifdefined\NINEPAGES
%\begin{figure}[]
\begin{figure*}[]
\else
\begin{figure*}[]
\fi
	%\medskip
	\hspace*{-0.4cm}
	\centering
	%\frame{
\ifdefined\NINEPAGES	
    \includegraphics[width=0.7\linewidth]{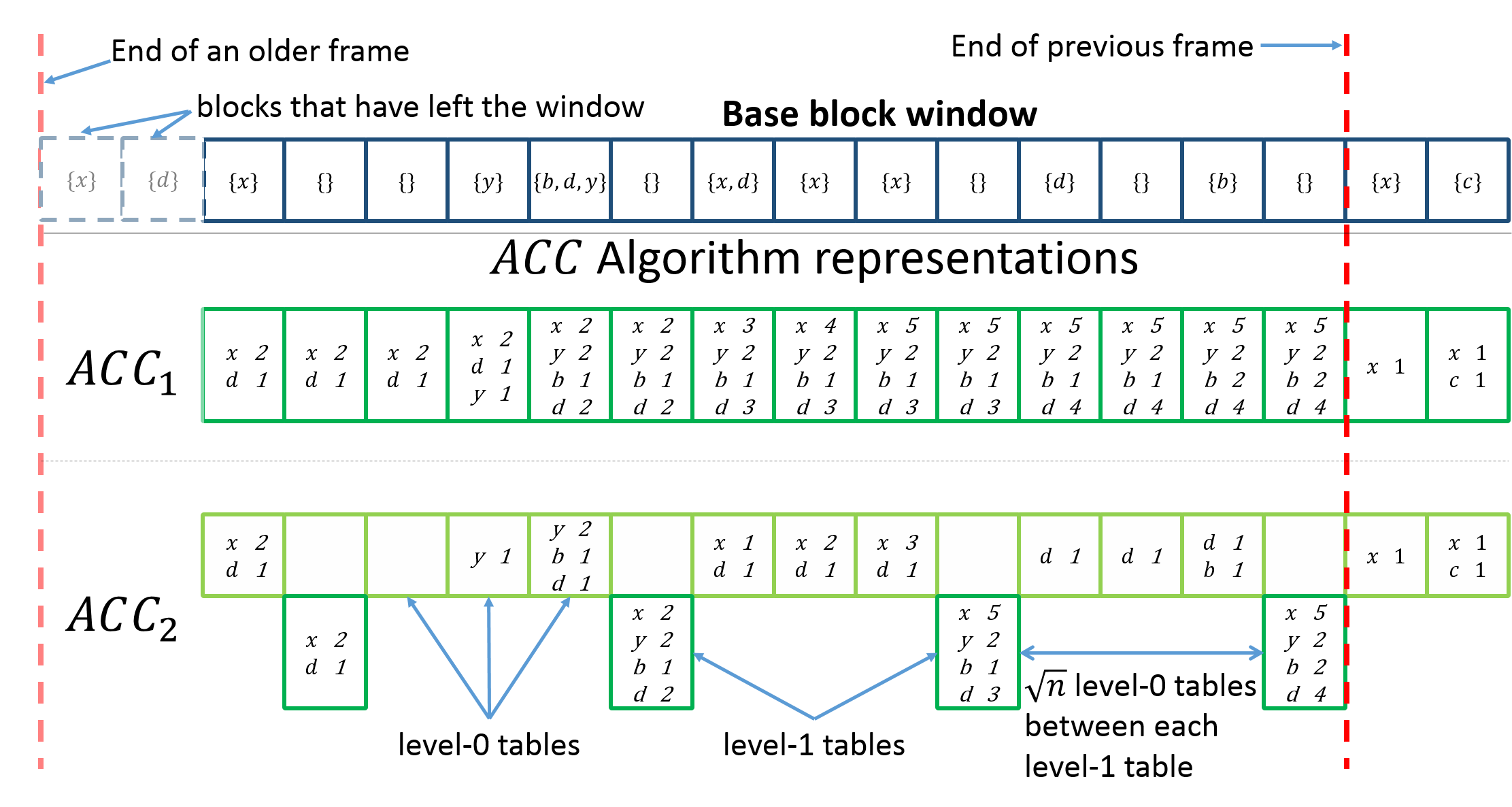}
\else
	\includegraphics[width=1.0\linewidth]{ACC_Illustraion.png}
\fi
	%}
	\caption{Illustration of the $\mathit{ACC_1}$ and $\mathit{ACC_2}$ algorithms. $\mathit{ACC_1}$ has only $\mathit{level}_0$ tables that track how many times each item has arrived from the beginning of the frame, while $\mathit{ACC_2}$ has two levels of tables. 
	Each $level_1$ table in $\mathit{ACC_2}$ tracks the frequencies from the beginning of the frame, while $level_0$ tables aggregate the data from the previous $level_1$ table.
	%Regarding $\mathit{ACC_2}$ tables, each frame is divided into segments, $level_0$ tables of $\mathit{ACC_2}$ count only items frequency within a segment, while $level_0$ tables count items frequency from the beginning of the frame at the end of each segment. For example, item $x$ arrived at the first block, $\mathit{ACC_1}$ table count its arrival in each table of the frame, while $\mathit{ACC_2}$ count its arrival in the corresponding block $level_0$ table and in each of $level_1$ tables.
 }
	\label{fig:ACC}
	\ifdefined\VLDB
	\vspace{-3mm}
	\fi
\ifdefined\NINEPAGES	
%\end{figure}
\end{figure*}
\else
\end{figure*}
\fi

$\mathit{ACC_2}$ saves space at the expense of additional table accesses.
Tables now have ``levels'', such that each table is either in $level_0$ or $level_1$.
The core idea is that $\mathit{ACC_1}$ is somewhat wasteful as it may create $O(n)$ table entries for each item, as it appears in all tables within the frame after its arrival.
Instead, we ``break'' each frame into $\sqrt{n}$ sized
%level-0
\emph{segments}.
At the end of each segment, we keep a single $level_1$ table that counts item frequencies from the beginning of the frame.
Since we can use these tables just as in $\mathit{ACC_1}$, we are left with computing the queried item frequency \emph{within a segment}.
This is achieved with a $level_0$ table, which we maintain for each block.
Alas, unlike the $level_1$ tables, $level_0$ tables only keep the block frequency counts from the \emph{beginning of the segment} the block belongs to.
Thus, each appearance of an item (within a specific block) can appear on all $\sqrt{n}$ $level_1$ tables, but on at most $\sqrt{n}$ $level_0$ tables and this reduce space consumption.
Compared with $\mathit{ACC}_1$, $\mathit{ACC}_2$ reduces the overall number of table entries \mbox{from $O(N\cdot n)$ to $O(N\cdot\sqrt{n})$.}

For an interval $[i,j]$, let $\mathit{block}_i$  and $\mathit{block}_j$ be the block numbers of $i$ and $j$ respectively.
To answer any interval frequency query of item $x$, we consider two cases:
If $\mathit{block}_i$ and $\mathit{block}_j$ are in the same frame, we access $\mathit{block}_i$'s and $\mathit{block}_j$'s tables to get $x$'s frequency from the beginning of the frame till $\mathit{block}_i$, $\mathit{block}_j$ and subtract the results, \changed{line \ref{line:fixed_result}} .
If $\mathit{block}_i$ and $\mathit{block}_j$ are in different frames, we consider $x$'s frequency in the $j$ blocks within the current frame by accessing $\mathit{block}_j$'s tables, plus its frequency within the last $i$ blocks of the previous frame.
To do so, we compute $x$'s frequency from the beginning of the previous frame, \changed{lines \ref{line:case2_begin} - \ref{line:case2_end}}.
A corner case that arises is that the $level_1$ table that includes $\mathit{block}_j$ may have already left the table.
%But the corresponding tables also count blocks that left the window, yet are part of the previous frame.
We solve it by maintaining \ghostTables{} for leaving segments: for $1\le \ell\le k$, \ghostTables$[\ell]$ contains the table of last leaving block that has a table at $level_\ell$, \changed{line \ref{line:ghost_table}}.
Hence, we can subtract the corresponding \ghostTables{} entries as well.

\begin{table}[]
	\scriptsize
%\small
\centering
	\begin{tabular}{|c|p{5.9cm}|}
		\hline
		% & Type \tabularnewline
		%\hline
		\BS & the number of segments from a level that consist of a next-level~segment.
		\tabularnewline
		\hline
		\Tables$[\ell,idx]$ & used for tracking block frequencies.
		Each table is identified with a level $\ell$ and the index of the last block in its segments.
		\tabularnewline
		\hline
		\incTables$[\ell]$ & tables for incomplete segments.
		\tabularnewline
		\hline
		\ghostTables$[\ell]$ & tables for leaving segments.
		\tabularnewline
		\hline
		\OFFSET & The offset within the current frame. \tabularnewline
		\hline
	\end{tabular}
	\normalsize
	\vspace{-.5em}
	\caption{Variables used by $\mathit{ACC_k}$ algorithm.}
	\label{tbl:accVars}
	\ifdefined\VLDB
	\vspace{-1.2em}
	\fi
\end{table}

Next, we generalize this to arbitrary $k$ values.
In $\mathit{ACC_k}$, we have $k$ levels of tables and segments.
We consider each block to be in its own $level_0$ segment and maintain a $level_0$ table for it.
Inductively, each $level_\ell$ segment (for $1\le \ell\le k$) consists of $n^{1/k}$ $level_{\ell-1}$ segments.
That is, each $level_1$ segment contains $n^{1/k}$ blocks, $level_2$ segments each consists of $n^{1/k}$ $level_1$ segments for a total of $n^{2/k}$ blocks, etc.
As each item may now appear in at most $n^{1/k}$ tables of each level, we get that the overall number of table entries is $O(N\cdot k\cdot n^{1/k})$.
To avoid lengthy computations at the end of each segment, we maintain $k$ additional ``incomplete'' tables that contain the cumulative counts for segments that already started, but not all of their blocks have ended yet.
A pseudo-code of the $\mathit{ACC_k}$ algorithm appears in Algorithm~\ref{alg:ACC}.

\ifdefined\NINEPAGES
\vspace*{-5mm}
\fi
\begin{algorithm}[t]
	\caption{$ACC_k$}\label{alg:ACC}
	\ifdefined\VLDB
	\scriptsize
	\fi
%\small
	%\begin{multicols}{2}
	\begin{algorithmic}[1]
		%		\ifdefined \NINEPAGES
		\Statex		Init: $%\BS \gets n^{1/k},
		\OFFSET\gets 1,
		\tableDis\triangleq {n^{1/k}},$\\
%		\tableDis\triangleq \floor{n^{1/k}},		
		\qquad{} $initialize\qquad \Tables, \incTables, \ghostTables$
		\Function {Add}{$x$}
		\For {$\ell \in 0,1,\ldots, k-1$}\Comment{Update all incomplete tables}
		\State $\incTables[\ell](x) += 1$
		\EndFor
		\EndFunction			
		\Function {EndBlock()}{}		
		\State $\ell \gets 0$
		\While {$((\ell < k - 1) \wedge \OFFSET \mod \tableDis^{\ell+1} = 0)$}
		\Statex\Comment{A Level-$\ell$ block has ended}
			\State $empty \ \incTables[\ell], \ghostTables[\ell]$ \Comment{Delete all entries}
			\State $\ell \gets \ell + 1$
%		\While {$((\ell\le k) \wedge (\OFFSET \mod n^{\ell/k} == 0))$} \Comment{End of a level $l$-segment}		
		\EndWhile
%		\If {$\OFFSET \mod \tableDis = 0$}
%			\State $\ell \gets 0$
%		\EndIf
		\State $\ghostTables[\ell] \gets \Tables[\ell,\OFFSET]$ \label{line:ghost_table}	
		\State $\Tables[\ell,\OFFSET] \gets \incTables[\ell]$\Comment{Copy Table}
		\If {$\OFFSET = n$} \Comment{New frame}
			\State $empty \ \incTables[k-1], \ghostTables[k-1]$
%			 \Comment{Reset all~entries}			
		\EndIf
		\State $\OFFSET \gets 1 + (\OFFSET \mod n)$		
		\EndFunction
		\Function {\windowQuery{}}{}\Comment{Frequency in the last $w$ blocks}\label{func:winQuery}
			\State \resizebox{7.45cm}{!}{{\scriptsize $\curFrameFreq\gets incTables[0](x) + \sum_{\ell=1}^{\log_{\tableDis}\OFFSET} \Tables\brackets{\ell, \tableDis^{\ell}\floor{\frac{\OFFSET}{\tableDis^{\ell}}}}(x)$}\small \label{line:sum_tables}}
			\If {$w\le \OFFSET + 1$}
				\State\Return \resizebox{6.1cm}{!}{{\scriptsize $\curFrameFreq - \sum_{\ell=0}^{\log_{\tableDis}\OFFSET+1-w} \Tables\brackets{\ell, \tableDis^{\ell}\floor{\frac{\OFFSET+1-w}{\tableDis^{\ell}}}}(x)$}} \label{line:fixed_result}
			\EndIf
			\State $B \gets n+\OFFSET+1-w$ \label{line:case2_begin}
			\State $L \gets \max\set{\ell\mid \tableDis^{\ell}\floor{\frac{B}{\tableDis^{\ell}}}\ge \OFFSET+1}$
			%\State {\tiny$preW \gets \sum_{\ell=0}^{L} \Tables\brackets{\ell, \tableDis^{\ell}\floor{\frac{B}{\tableDis^{\ell}}}}(x) + \sum_{\ell=L+1}^{k-1} \ghostTables[\ell](x)$}
			\State {\tiny$preW \gets \sum_{\ell=0}^{L} \Tables\brackets{\ell, \tableDis^{\ell}\floor{\frac{B}{\tableDis^{\ell}}}}(x) + \sum_{\ell=L+1}^{k-1} \ghostTables[\ell](x)$}
			\small
			\State\Return $\curFrameFreq+\Tables[k-1,n](x) - preW$ \label{line:case2_end}
%			-sum_{\ell=0}^{\log_{\tableDis}\OFFSET-w} \Tables\brackets{\ell, \tableDis^{\ell}\floor{\frac{\OFFSET-w}{\tableDis^{\ell}}}}(x) - sum_{\ell=L+1}^{k} \ghostTables[\ell](x)$
		\EndFunction
		\Function {\SIQ}{$x,i,j$}		
		\If {$i = 0$} \label{line:case1}
			\State \Return \windowQueryText{}$(x,j)$
		\EndIf
		\State \Return \windowQueryText{}$(x,j)$ - \windowQueryText{}$(x,i)$
		\EndFunction		
{}	\end{algorithmic}
	%\end{multicols}
	\normalsize
\end{algorithm}
\ifdefined\NINEPAGES
\setlength{\textfloatsep}{14pt}
\setlength{\intextsep}{16pt}
\setlength{\dbltextfloatsep}{12pt}
\setlength{\abovedisplayskip}{3pt}
\setlength{\belowdisplayskip}{3pt}
\fi
\subsubsection{Analysis}
%We analyze the ACC algorithm below.
%We start by bounding its memory consumption.
%In the next section, we discuss how this can be asymptotically improved for large identifiers (where $\mathcal U = n^{\omega(1)}$).

The following theorem bounds the memory consumption of the ACC algorithms.

\ifdefined\EXTENDED
\begin{theorem}
Denote the sum of cardinalities of the last $n$ blocks by $N$. Algorithm~\ref{alg:ACC} requires $O(N\cdot n^{1/k}\cdot k\cdot (\log n + \log|\mathcal U|))$ space.
\end{theorem}
\else
\begin{theorem}
	Denote the sum of cardinalities of the last $n$ blocks by $N$. Algorithm~\ref{alg:ACC} requires~$O(N n^{1\over k} k \logp{n|\mathcal U|})$~space.
\end{theorem}
\begin{proof}
    Each item may appears in at most $n^{1/k}$ tables of each level (because, as described before, each $level_\ell$ segment (for $1\le \ell\le k$) consists of $n^{1/k}$ $level_{\ell-1}$ segments).
   $\mathit{ACC_k}$ algorithm has $k$ levels of tables, thus each item may appear in total at most $kn^{1/k}$ tables.
    Each table entry consists of an $O(\log|\mathcal U|)$-bits identifier and a counter of $O(\log n)$ bits.
    Thus, the space of each item is $O(n^{1\over k} k \logp{n|\mathcal U|})$; hence, for $N$ items, $\mathit{ACC_k}$ requires~$O(N n^{1\over k} k \logp{n|\mathcal U|})$~bits.
\end{proof}
\fi
%\textbf{Ran: the correctness analysis of ACC seems to be missing. I think it should come here (perhaps with a proof at the full version). We also forgot to mention the ghost tables and their functionalities and discuss the query procedure. While the intuition behind ACC is described perfectly, it'll be hard to impossible to understand the pseudo code without this.}
\ifdefined\NINEPAGES
\fi

\begin{theorem}
	Algorithm \ref{alg:ACC} solves the \SIProblem{} problem.
\end{theorem}

\begin{sproof}
	We need to prove that upon an \SIQ{}
	${(x,i,j)}$ query, for any $i\le j\le n$ and $x\in\mathcal U$, $\mathit{ACC_k}$ is able to compute the \emph{exact} answer.
    Notice that in handling queries in Algorithm \ref{alg:ACC}, we split the computation in two:
    The first case is when $i=0$, Line \ref{line:case1}, this means that the end of the interval is also the last block, and thus we only need to return the frequency in the last $j$ block in the window, as calculated by \windowQueryText{}$(x,j)$.
    Otherwise, we subtract the frequency that is calculated \windowQueryText{}$(x,i)$ from the result of \windowQueryText{}$(x,j)$.
    Hence, we need to show that the frequency calculated by \windowQueryText{}$(x,w)$ is correct.

    As is evident from the code in Lines $1$--$13$, $\incTables$ store the frequency of items within the current block, while $\Tables$ store the frequencies of completed blocks from the beginning of their frame.
    Consider the case where the entire interval is within the current frame.
    In this case, the frequency of an item in the last $w$ blocks can be calculated as its frequency in the current block (using $\incTables$) plus its frequency in the preceding blocks, as is done in Line \ref{line:sum_tables}, and stored in $\curFrameFreq$.
    Notice that to reduce query time, we access the highest level containing this information.
    However, since $\Tables$ store the frequency from the beginning of the frame, we need to subtract from $\curFrameFreq$ the item's frequency in prior blocks, which is done in Line \ref{line:fixed_result} (again, by accessing the highest level tables that include this data).

    The second case is when the given interval crosses into the previous frame.
    In this case, we need to add to $\curFrameFreq$ the frequency of the blocks that are included in the previous frame.
    Once again, we need to query the table holding the frequency in the last relevant block of that frame and subtract from the result the frequency in the preceding tables.
    As some of these tables might be beyond an entire window limit, their information might be stored in $\ghostTables$ rather than $\Tables$.
    This is handled in Lines~\ref{line:case2_begin}--\ref{line:case2_end}.
%    we simply need to return $\ensuremath{g_x[j]}$.
%	Let $\ensuremath{g_x[w]}$ be the frequency of item $x$ in the last $w$ blocks.
%	If $i=0$, we simply need to return $\ensuremath{g_x[j]}$.
%	Otherwise, $\ensuremath{g_x[j]}$ - $\ensuremath{g_x[i]}$.
%	The function at Line \ref{func:winQuery} computes $\ensuremath{g_x[w]}$ by considering two cases: If $w$ is within the current frame, then we use incTables for frequency at $level_0$ and $\Tables$ for the other levels. Otherwise, If $w$ is within the previous frame, we use in addition the $\ghostTables[\ell](x)$.
\end{sproof}

\ifdefined\EXTENDED
\begin{sproof}
	We need to prove that upon an \SIQ{}${(x,i,j)}$ query, $\mathit{ACC_k}$ is able to compute the \emph{exact} answer for any $i\le j\le n$ and $x\in\mathcal U$.
    Notice that in handling such queries in Algorithm \ref{alg:ACC}, we split the computation in two:
    If $i=0$, this means that the end of the interval is also the last block, and thus we only need to return the frequency w.r.t. the last $j$ block in the window, as calculated by $\windowQueryText{}(x,j)$.
    Otherwise, we subtract the frequency that is calculated $\windowQueryText{}(x,i)$ from the result of $\windowQueryText{}(x,j)$.
    Hence, we need to show that the frequency calculated by $\windowQueryText{}(x,w)$ is correct.

    As evident from the code in Lines $1$--$13$, $\incTables$ store the frequency of items within the current block while $\Tables$ store the frequencies of items of completed blocks from the beginning of their frame.
    Consider the case where the entire range is within the current frame.
    Hence, in principle, the frequency of an item in the last $w$ blocks can be calculated as its frequency in the current block ($\incTables$) plus its frequency in the preceding blocks, as is done in Line~$15$, and stored in $\curFrameFreq$.
    Notice that to reduce query time, we access the highest level containing this information.
    However, since $\Tables$ store the frequency from the beginning of the frame, we need to subtract from $\curFrameFreq$ obtained in Line~$15$ the frequency of this item in prior blocks, which is done in Line~$17$ (here again, accessing the highest level tables that include this data).

    The second case is when the range crosses into the previous frame.
    In this case, we need to add to $\curFrameFreq$ the frequency of the blocks that are included in the previous frame.
    Once again, we need to find the table holding the frequency in the last relevant block of that frame, and subtract from it the frequency in the preceding tables.
    Yet, as some of these tables might be beyond an entire window limit, their information might be stored in $\ghostTables$ rather than $\Tables$.
    This is handled in Lines~$20$--$21$.
%    we simply need to return $\ensuremath{g_x[j]}$.
%	Let $\ensuremath{g_x[w]}$ be the frequency of item $x$ in the last $w$ blocks.
%	If $i=0$, we simply need to return $\ensuremath{g_x[j]}$.
%	Otherwise, $\ensuremath{g_x[j]}$ - $\ensuremath{g_x[i]}$.
%	The function at Line \ref{func:winQuery} computes $\ensuremath{g_x[w]}$ by considering two cases: If $w$ is within the current frame, then we use incTables for frequency at $level_0$ and $\Tables$ for the other levels. Otherwise, If $w$ is within the previous frame, we use in addition the $\ghostTables[\ell](x)$.
\end{sproof}
\else
%The proof is technical and follows directly from the code.
%It is therefore delayed to the full version of this paper~\cite{full-version}.
\fi

\subsection{Hierarchical Interval Tree ($\mathit{HIT}$)}
\label{sec:hit}

\ifdefined\NINEPAGES
%\begin{figure}[]
\begin{figure*}[]
	\else
	\begin{figure*}[]
		\fi
		%\medskip
		\hspace*{-0.4cm}
		\centering
		%\frame{
		\ifdefined\NINEPAGES	
        \includegraphics[width=0.7\linewidth]{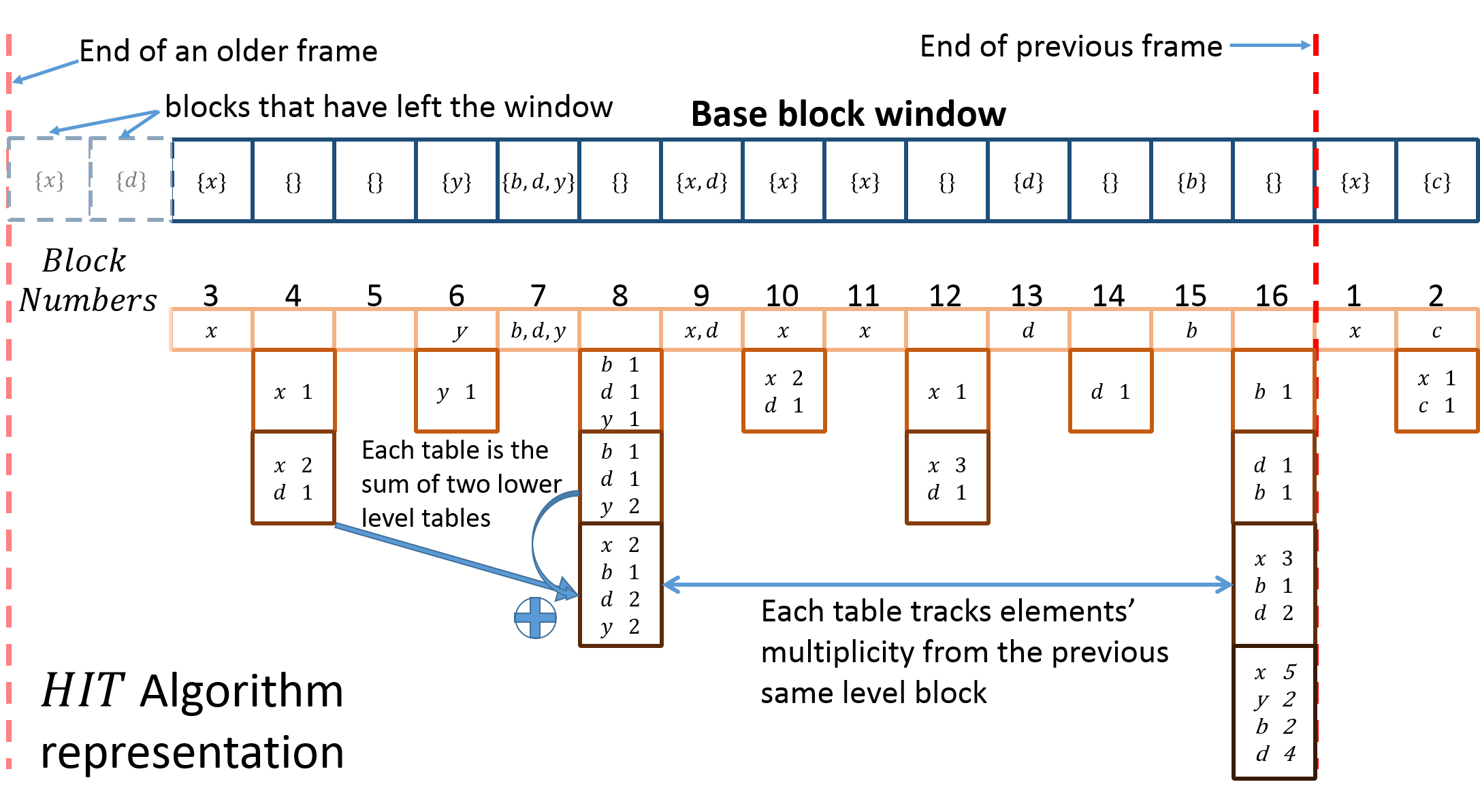}
		\else
		\includegraphics[width=1.00\linewidth]{HIT_Illustraion.png}
		\fi
		%}
%		\caption{\vspace*{-1em}HIT}
		\caption{$\mathit{HIT}$ algorithms, first level tables track how many times each item arrived within the corresponding block. At $\mathit{level}_\ell$, tables of $\mathit{block}_i$ track how many times each item has arrived \emph{between $\mathit{block}_{i - 2^{\ell} + 1}$  and $\mathit{block}_i$}. For example, item $b$ arrives once at $\mathit{block}_7$, so $\mathit{block}_7$ $\mathit{level}_0$ table contains $b$ with count $1$, $\mathit{block}_8$ $\mathit{level}_1$ table tracks how many times each item arrived between $\mathit{block}_7$ and $\mathit{block}_8$, so it contains $b$ with count $1$ and $\mathit{level}_2$ table for $\mathit{block}_8$, track how many times each item arrived between $\mathit{block}_5$ and $\mathit{block}_8$, and as well contains $b$ once.
		Table at $\mathit{level}_{\ell + 1}$ merges two $\mathit{level}_\ell$ frequency tables. For example, $\mathit{block}_8$ third level table merges second level tables of $\mathit{block}_8$ and $\mathit{block}_8$.}
		\label{fig:HIT}
		\ifdefined\NINEPAGES	
		\vspace{-3mm}
	%\end{figure}
	\end{figure*}
	\else
\end{figure*}
\fi

Hierarchical Interval Tree, denoted $\mathit{HIT}$, tracks flow frequencies using a hierarchical tree structure in which each node stores the partial frequency of its sub-tree.
%
%solves the block interval frequencies problem using a hierarchical tree of frequency tables.
%%This algorithm accesses $O(\log \oneOverE)$ tables at most in interval query and update operations.
%The HIT algorithm breaks the stream into frames, each of size $n$, where $n$ is the maximal window size, and within each frame uses a hierarchical tree of tables to compute elements' frequencies.
%
%The hierarchical tree consists of levels, for every level in the tree, we break each frame into blocks.
%That is, level-0 contains $n$  \emph{blocks}, level-1 contains $\frac{n}{2}$ \emph{blocks}, etc.
%Inductively, level-$\ell$ contains $\frac{n}{2^{\ell}}$ \emph{blocks}.
%
Precisely, the levels of the tree are defined as follows:
$\mathit{level}_0$ includes frequency tables, one for each block of the stream, that track how many times each item arrived \emph{within the corresponding block}.
Tables at $\mathit{level}_\ell$ of $\mathit{block}_i$ track how many times each item has arrived \emph{between $\mathit{block}_{i - 2^{\ell} + 1}$  and $\mathit{block}_i$}, where $0 < \ell \le trailing\_zeros(i)$, \changed{line \ref{line:hit_tables}}.
That is, these tables contain partial queries results for each item and track item's multiplicity from the previous same level block.
Hence, each level contains tables for half the blocks of the previous level, and thus each $\mathit{block}_i$ has tables in \textit{trailing\_zeros}$(i)$ levels; we assume that the number of trailing zeros can be computed efficiently with the $ctz$ machine instruction in modern CPUs.
An illustration of the algorithm appears in Figure~\ref{fig:HIT}.

For example, consider $\mathit{block}_9$, $\mathit{block}_{10}$, $\mathit{block}_{11}$ and $\mathit{block}_{12}$ in Figure~\ref{fig:HIT}.
During $\mathit{block}_9$, items $x$ and $d$ arrive; $x$ also arrives in $\mathit{block}_{10}$ and $\mathit{block}_{11}$, while there are no items arrivals in $\mathit{block}_{12}$.
So the tables of $\mathit{block}_{12}$ will be as follow: $\mathit{level}_0$ table is empty because there is no items arrival within $\mathit{block}_{12}$.
$\mathit{level}_1$ table tracks items arrival between $\mathit{block}_{11}$ and $\mathit{block}_{12}$; its content will be item $x$ with count $1$.
$\mathit{level}_2$ table counts the item arrival between $\mathit{block}_9$ and \changed{$\mathit{block}_{12}$,} so it will contain item $x$ three times ($\mathit{block}_9$, $\mathit{block}_{10}$ and $\mathit{block}_{11}$) and $d$ once (in $\mathit{block}_9$).
%\TODO: mention efficient implementation of ctz
Note that each table at $\mathit{level}_{\ell + 1}$ merges two $\mathit{level}_\ell$ frequency tables.

We can compute any interval frequency by using the hierarchical tree tables.
While this can be done using linear scan, the higher levels of the tree are designed to allow efficient time computation by using the stored partial queries.

Notice that some of the partial queries results stored in the higher levels may be invalid.
For example, in case a new block is added, the oldest one departs the window, so the content of tables that refer to the departing block become invalid.
We solve this problem by choosing the levels to use such that we only consider valid tables.
Let $\mathit{block}_i$  and $\mathit{block}_j$ be the block numbers of the first interval index and the second one.
Here, we scan backward from $\mathit{block}_j$ to $\mathit{block}_i$, greedily using the highest possible level at each point, \changed{line \ref{line:greedy_choose}}.
This minimizes the number of needed steps.
If $\mathit{block}_j> block_i$, all tables along the way are valid.
In this case, we only need $\log_2(block_j - block_i +1)$ value look-ups.
Otherwise, we choose $\mathit{level}_0$ tables between blocks $1$ and $\mathit{block}_j$, so we need $\log_2(block_j + 1)$ value look-ups, and then another $\log_2(n-block_i+1)$ look-ups for querying the remaining interval.
Overall, our \mbox{computation takes at most $2\logn$ steps.}

We use an incremental table for incomplete blocks in each we increment an element $x$'s entry for any \text{ADD}$\bm{(x)}$ operation, \changed{line \ref{line:incremental_table}}.
The pseudo code of the algorithms appears in Algorithm~\ref{alg:HIT} and Table~\ref{tbl:sliding-window-vars} contains a list of the used~variables.

\begin{algorithm}[t]
	%\caption{ֿ$\mathit{HIT}$}\label{alg:HIT}
    \caption{$\mathit{HIT}$}\label{alg:HIT}
    \ifdefined\NINEPAGES
	\scriptsize
	\fi
%\small
	%\begin{multicols}{2}
	\begin{algorithmic}[1]
		%		\ifdefined \NINEPAGES
		\Statex		Initialization: $%\BS \gets n^{1/k},
		\OFFSET\gets 0, initialize\qquad \mathit{\Tables, \incTable}.$
		\Function {Add}{$x$}
		%\For {$\ell \in 0,1,\ldots, ctz(\OFFSET)$}\Comment{Update all incomplete tables}
		%\State $\incTables[\ell](x) += 1$
		\State $\incTable(x) += 1$\Comment{Update the incomplete block's tables} \label{line:incremental_table}
		%\EndFor
		\EndFunction			
		\Function {EndBlock()}{}		
		\State $\OFFSET \gets (\OFFSET + 1) \mod n$\label{line:offset}
        \State $\Tables[0,idx] \gets \incTable$
        \State $empty \quad \incTable$ \Comment{Delete all entries.}
        \For {$\ell \in 1,\ldots,ctz(\OFFSET)$}
		%\State $\ell \gets 1$
		%\While {$((\ell\le  ctz(\OFFSET))$} \Comment{If $\OFFSET$ has at least $\ell$ trailing zeros}
		%\State $\Tables[\ell,idx] \gets \incTables[l]$
		%\If {$\ell \ge 1$}
		\ifdefined\VLDB
		\State
        {{$\Tables[\ell,\OFFSET] =\Tables[\ell - 1,\OFFSET] $\label{line:hit_tables}\Statex\hspace{3.5cm} $+ \Tables[\ell - 1,\OFFSET - 2^{\ell} + 1]$}} %\small
        \else
        \State
        {{$\Tables[\ell,\OFFSET] =\Tables[\ell - 1,\OFFSET] + \Tables[\ell - 1,\OFFSET - 2^{\ell} + 1]$}}\label{line:hit_tables}
        \fi
		%\vspace*{-.6cm}
		%\begin{multline*}
		%\hspace*{0.8cm}\Tables[\ell,\OFFSET] =\Tables[\ell - 1,\OFFSET] \\+ \Tables[\ell - 1,\OFFSET - 2^{\ell} + 1]
		%\end{multline*}
		%\EndIf
		%\State $\ell \gets \ell + 1$
		%\EndWhile
        \EndFor
		\EndFunction
		\Function {\SIQ}{$x,i,j$}		
		\State $last = (\OFFSET - j) \mod n$ \Comment{The most recent block's index}
		\State $first = (\OFFSET - i) \mod n$ \Comment{The oldest queried block's}
		\State $ b \gets\mbox{$first$}$
		\State $count \gets 0$
		\State $ d \gets\mbox{$1 +\parentheses{first - last \mod n}$}$
		\While { $d > 0$} \label{line:while}
		\State $ level \gets\mbox{$\min({ctz(b), \floor{\log{d}}})$}$ \label{line:greedy_choose}
		\State $ count \gets\mbox{count + \Tables[level, b](x)}$\label{line:count}
		\State $d \gets\mbox{$d$ - $2^{level}$}$\label{line:step}
		\State $b \gets\mbox{$b$ - $2^{level}$}$\label{line:block}
		\If {$ b = 0$}
		\State $ b \gets k$
		\EndIf
		\EndWhile		
		\State\Return {$count$}
		\EndFunction
	\end{algorithmic}
	%\end{multicols}
	\normalsize
\end{algorithm}

\begin{table}[t]
\ifdefined\VLDB
	\scriptsize
\fi
	\centering
	\begin{tabular}{|c|p{6.0cm}|}
		\hline
		% & Type \tabularnewline
		%\hline
		\BS & Number of blocks from a level that consist of a next-level~block.
		\tabularnewline
		\hline
		\Tables$[\ell,idx]$ & used for tracking block frequencies.
		Each table is identified with a level $\ell$ and the index of the last block in its block.
		\tabularnewline
		\hline
		\incTable & A table for the most recent, incomplete, block.
		\tabularnewline
		\hline
		\OFFSET & The offset within the current frame. \tabularnewline
		\hline
	\end{tabular}
	\normalsize
	\ifdefined\VLDB
	\vspace{-.7em}
	\fi
	\caption{Variables used by $\mathit{HIT}$ algorithm.}
	\label{tbl:sliding-window-vars}
\ifdefined\NINEPAGES	
	\vspace{-1.3em}
\fi
\end{table}

\subsubsection{Analysis}
\ifdefined\EXTENDED
We now analyze the $\mathit{HIT}$ algorithm.
We start by proving its correctness.
\fi

\begin{theorem}
Algorithm \ref{alg:HIT} solves the \SIProblem{} problem.
\end{theorem}

\begin{proof}
We need to prove that upon an \SIQ{}${(x,i,j)}$ query, for any $i\le j\le n$ and $x\in\mathcal U$, $\mathit{HIT}$ is able to compute the \emph{exact} answer  without error.
We first introduce some notations.
$x$ denotes the queried element; $\ensuremath{f_x[i]}$ indicates the frequency of item $x$ during the $i^{th}$ block, so that the newest block's index is $1$.
According to Line \ref{line:while}, we iterate over blocks, query $\mathit{level}_\ell$ table, which contains partial query of $2^{\ell}$ blocks, thus we can advance by $2^{\ell}$, Lines \ref{line:step} and \ref{line:block}.
The output of the algorithm for \mbox{querying $x$ in interval $i\le j\le n$ is:}
\ifdefined\NINEPAGES
{ %\scriptsize
	\begin{multline}
	\SIQ{}{(x,i,j)} = \sum_{i=last}^{first} \ensuremath{f_x[i]}
	= \sum_{i=1}^{first} \ensuremath{f_x[i]} - \sum_{i=1}^{last} \ensuremath{f_x[i]}\\
	= \sum_{i=1}^{(\OFFSET - j) \mod n} \ensuremath{f_x[i]} - \sum_{i=1}^{(\OFFSET - i) \mod n} \ensuremath{f_x[i]}. \label{eq_proof}
	\end{multline}\vspace*{-0.3cm}
}
\else
\begin{multline}
\SIQ{}{(x,i,j)} = \sum_{i=last}^{first} \ensuremath{f_x[i]}
= \sum_{i=1}^{first} \ensuremath{f_x[i]} - \sum_{i=1}^{last} \ensuremath{f_x[i]}
= \sum_{i=1}^{(\OFFSET - j) \mod n} \ensuremath{f_x[i]} - \sum_{i=1}^{(\OFFSET - i) \mod n} \ensuremath{f_x[i]}. \label{eq_proof}
\end{multline}
\fi

According to the definition of \xSetWindowFrequency[k] in section \ref{block_interval_definition}, we got that \eqref{eq_proof} is equal to:
$\xSetWindowFrequency[j] - \xSetWindowFrequency[i] \triangleq \xSetIntervalFrequency.$
\end{proof}
%\textbf{Ran: the space consumption analysis seems to be missing. We can state it here and prove in the full version.}
%The following's proof is deferred to the full version~\cite{full-version}.
\begin{theorem}
	Denote the sum of cardinalities of the last $n$ blocks by $N$. Algorithm~\ref{alg:HIT} requires~$O(N\log{n} \logp{n|\mathcal U|})$~space.
\end{theorem}
%\ifdefined\EXTENDED
\begin{proof}
	As described above, each element's appearance may reflect in $O(\log{n})$ tables.
	Every table entry takes $O(\log{|\mathcal U|})$ bits for the key and another $O(\log{n})$ for the value, and thus the overall space is $O(N\log{n}\logp{n|\mathcal U|})$.\qedhere
	
%	$$O(N(\log n \logp{\min\set{n,|\mathcal U|}} + \log |\mathcal U|) ).$$
\end{proof}
%\fi 

\ifdefined\VLDBREVISION
\vspace{-5mm}
\changed{
\paragraph*{Optimizations}
In the full version~\cite{full-version} we provide optimizations that reduce HIT's space to $O(N(\log|\mathcal{U}|+\log{n}\log{N}))$ and that of $ACC_k$ to $O(N (\log|\mathcal{U}|+n^{1\over k} k \log{N}))$. We also discuss how to \emph{deamortize} the update process to get the worst-case time complexity \mbox{equivalent to the amortized analysis.}
}
\else
\subsection{Optimizations}

\ifdefined\EXTENDED
This section includes optimizations that can be applied to the $ACC_K$ and/or $HIT$ algorithms.
\else
Here we discuss how to optimize our algorithms.
\fi
\paragraph*{Short IDs}
\ifdefined\EXTENDED
Element IDs are often quite long (alternatively, $\mathcal U$ is large), e.g., a 5-tuple identifier per flow may take over 100 bits and Internet URLs can be even longer.
Hence, when the size of item IDs are large, we can reduce their required space as follows:
For each frame, we maintain an $O(N)$ sized array of items identifiers that were added to some block during the frame.
Every time a new (distinct) item arrives, we add it to the array.
To find the index of each ID in the array, we maintain an additional table that maps IDs to their array indices.
Clearly, the combined space requirement of the array and map table is $O(N\cdotpa{\log\mathcal{U} + \log N})$.
Finally, we replace the keys in the algorithms' tables (at all levels) such that instead of storing identifiers we use the array indices as keys.
Given a query, we first find the array index using the new table and then follow the same procedure as before, but with the index as key.
This optimization can be applied to both $ACC_K$ and $HIT$.
\else
Element IDs are often long; e.g., a 5-tuple identifier may take over 100 bits while Internet URLs can be even longer.
When the size of item IDs are large, we can reduce their required space as follows:
For each frame, we maintain an $O(N)$-sized array that contains the identifiers of items that were added to some block.
When a new (distinct) item arrives, we add it to the array.
To find the index of each ID in the array, we maintain an additional table that maps IDs to their array indices.
Clearly, the combined space requirement of the array and map table is $O(N\cdotpa{\log|\mathcal{U}| + \log N})$.
Finally, we replace the keys in the algorithms' tables (at all levels) such that instead of storing identifiers we use the array indices as keys.
Given a query, we first find the array index using the new table and then follow the same procedure as before, but with the index as key.
%This optimization can be applied to both $ACC_k$ and $HIT$.
\fi
Thus, we always store at most $O(N)$ IDs.
This reduces HIT's space to $O(N(\log|\mathcal{U}|+\log{n}\log{N}))$ and that of $ACC_k$ to $O(N (\log|\mathcal{U}|+n^{1\over k} k \log{N}))$.
%\paragraph*{Variable-Length Counters}
%Since a level-$\ell$ table in HIT sums element frequency over $2^\ell$ blocks, we can use $\ell$ bits for its value instead of $\log n$. Thus, the overall space reduces to:
%$$N\sum_{\ell=0}^{\log n}O(\log |U|+\ell) = O(N\log n \log |\mathcal U|).$$

\paragraph*{Deamortization}
Algorithm \ref{alg:reduction} shows a reduction from the \SIProblem{} problem to \IFQ{}.
For any $\mathbb A$ algorithm that solves \SIProblem{}, we notice that the operation $\mathbb A.Add(x)$ cannot be called more than once a block for the same element $x$.
We can then \emph{deamortize} the $\mathbb A.Add(x)$ operation for reducing the worst case update time.
Namely, we can spread the time required for $Add(x)$ over an entire block. This means that if queried for $x$, we may miss $1$ from its block frequency. Nevertheless, this only adds an error of $\blockSize$ which can be compensated for by slightly \mbox{reducing the block~size.}
%Thus, when an item arrives, we can perform a few $Add()$ operations for elements that satisfy Line~\ref{line:SSoverflow}.%TODO
%Since there can be at most $O(\epsilon^{-1})$ overflows, we can spread the $Add()$ operation
\fi

\ifdefined \NINEPAGES

\fi
% \begin{table*}[t]
% 	\addtolength{\tabcolsep}{-0pt} \addtolength{\parskip}{-0.1mm} \center{
% 		\begin{tabular}{|c|c|c|c|c|c|c|c|}
% 			\hline
% 			Trace &Date(Y/M/D) &  \#Packets & Total volume & Mean size & Max size & \% large packets  & \% large packet traffic\tabularnewline
% 			\hline
% 			\hline
% 			Chicago16 & 2016/02/18 & 97 M & 94 GB & 1046 & 49458 & 0.34\% & 0.5\%\tabularnewline
% 			\hline
% 			Chicago15 & 2015/12/17 & 85 M & 80 GB & 1013 & 64134 & 0.22\%  & 0.34\%\tabularnewline
% 			\hline
% 			SanJose14 & 2014/06/19 & 112 M & 149 GB & 1424 & 65535 & 0.78\% & 25.02\%\tabularnewline
% 			\hline
% 			SanJose13 & 2013/12/19 & 97 M & 110 GB & 1225 & 65528 & 0.49\% & 18.81\% \tabularnewline
% 			\hline
% 		\end{tabular}
% 		
% 		{{\caption{A summary of key characteristics of the real Internet traces used in this work. }
% 				\label{tbl:traces}
% 			}}}
% 		\end{table*}
\normalsize

\ifdefined\VLDBREVISION
\begin{figure*}[t]
	\center{
\ifdefined\NINEPAGES		
		\vspace*{-0.4cm}
\fi
		\begin{tabular}{cc}
			\subfloat{\includegraphics[width=\matrixCellWidth]
				{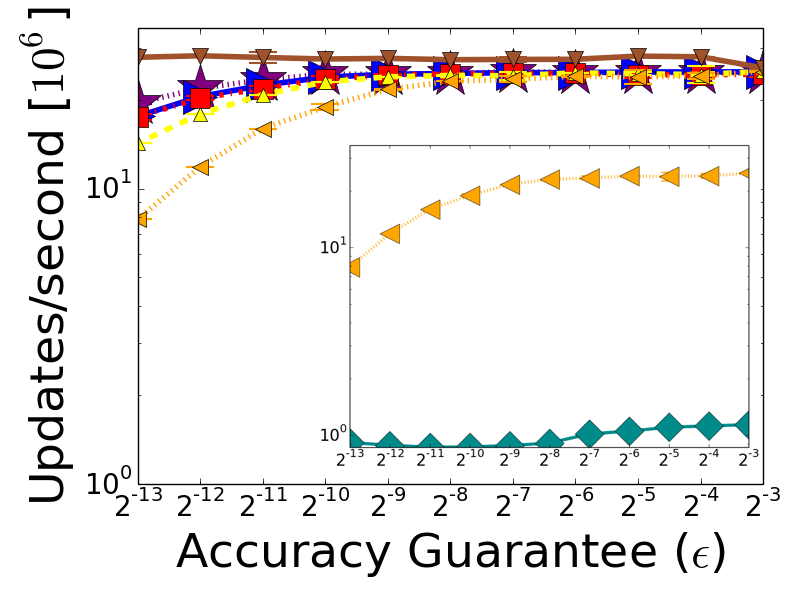}} &
			\subfloat{\includegraphics[width=\matrixCellWidth]{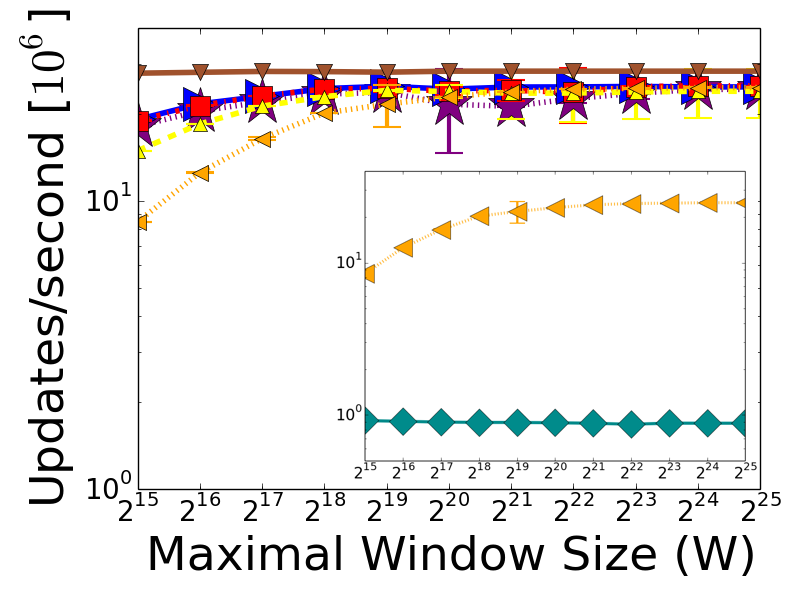}}
			\subfloat{\includegraphics[width=\matrixCellWidth]{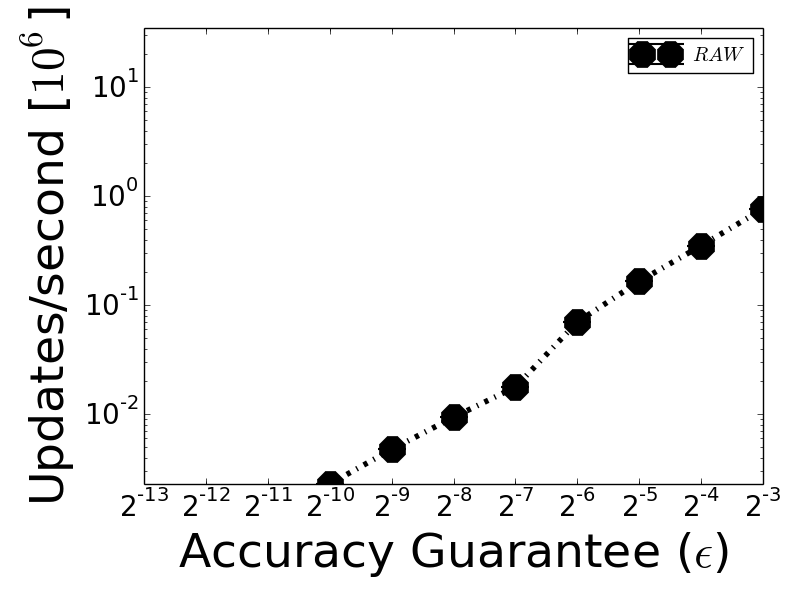}
			}
				\ifdefined\NINEPAGES
	\vspace*{-3mm}
	\fi
			\tabularnewline
            \multicolumn{2}{c}{\subfloat{\includegraphics[width = 12cm]
		{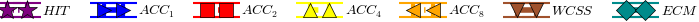}}}		
	\ifdefined\NINEPAGES
	\vspace*{-3mm}
	\fi
		\end{tabular}
		}
	\caption{\changed{Update operation runtime comparison as a function of the accuracy guarantee ($\epsilon$) and the maximum window size ($W$). Since the update speed of RAW is significantly slower than the other algorithms, we have placed it in a separate graph and we managed to run $\mathit{RAW}$ only up to $\epsilon = 2^{-10}$ due to its memory consumption limitation. The graphs plot contains a subplot which compares $\mathit{ECM}$ with our slowest algorithm, $\mathit{ACC_8}$ in this case. We see that $\mathit{ECM}$ is much slower than $\mathit{ACC_8}$. We also compared our algorithms with $\mathit{WCSS}$ which is the state of the art for the simpler problem \mbox{of a \emph{fixed} sliding window}.}}
	\label{fig:update}
\end{figure*}
\else

\begin{figure*}[t]
	\begin{tabular}{ccc}
		\subfloat[\changed{Backbone}]{\includegraphics[width=\matrixCellWidth]
			{IFQGraphs/epsilon/Update/NChicago16_Updates_3-14_new_epsilon.png}} &
		\subfloat[Datacenter]{\includegraphics[width = \matrixCellWidth]
			{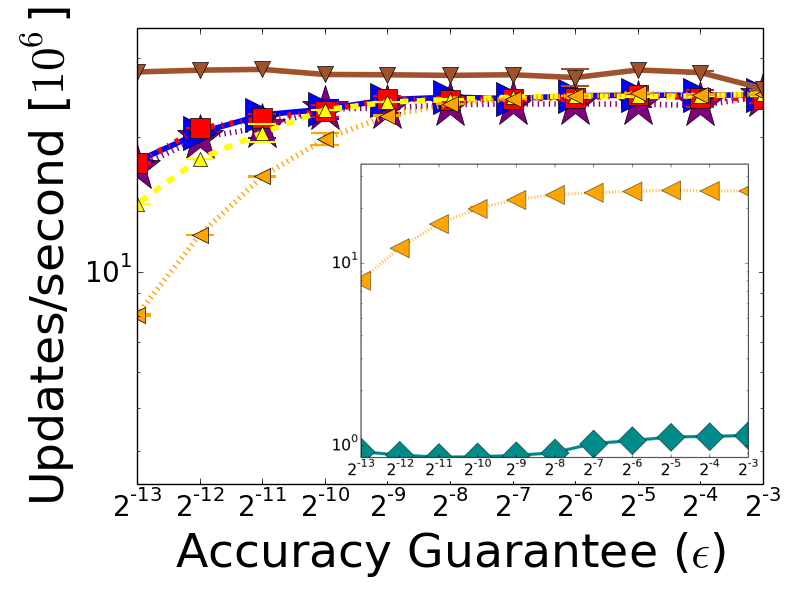}} &
		\subfloat[Edge]{\includegraphics[width = \matrixCellWidth]
			{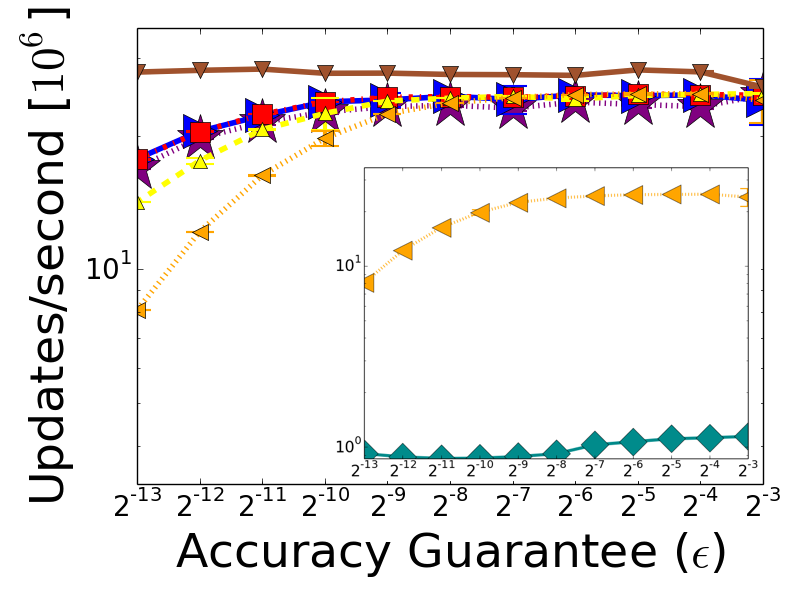}}
\ifdefined\NINEPAGES		
\\ 	[-4ex]
\else
\\
\fi				
\tabularnewline
\multicolumn{3}{c}{\subfloat{\includegraphics[width = 15cm]
		{IFQGraphs/epsilon/legend.png}}}		
\ifdefined\NINEPAGES		
\\ 	[-2ex]
\else
\\
\fi
\addtocounter{subfigure}{-1}		
		\subfloat[Backbone]{\includegraphics[width = \matrixCellWidth]
			{IFQGraphs/window/Update/NChicago16_updates_15-25without_ecm_new_window.png}} &
		\subfloat[Datacenter]{\includegraphics[width = \matrixCellWidth]
			{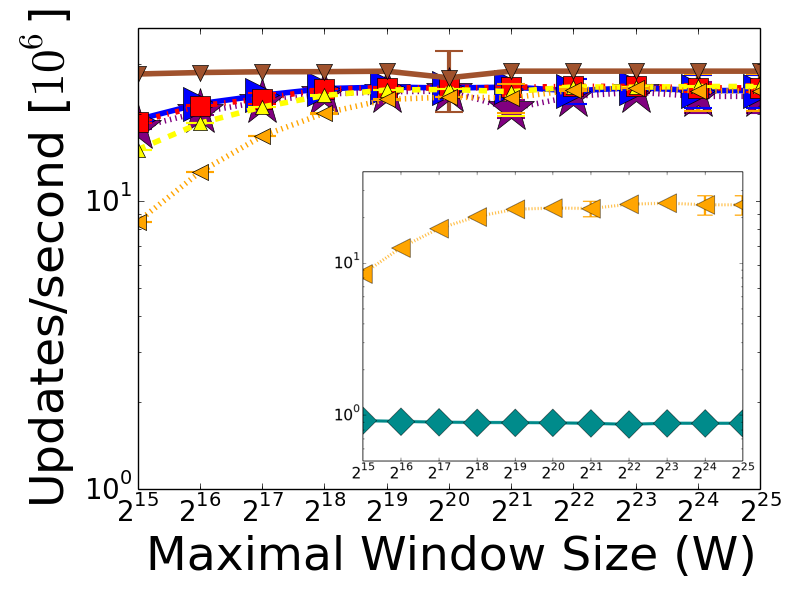}} &
		\subfloat[Edge]{\includegraphics[width = \matrixCellWidth]
			{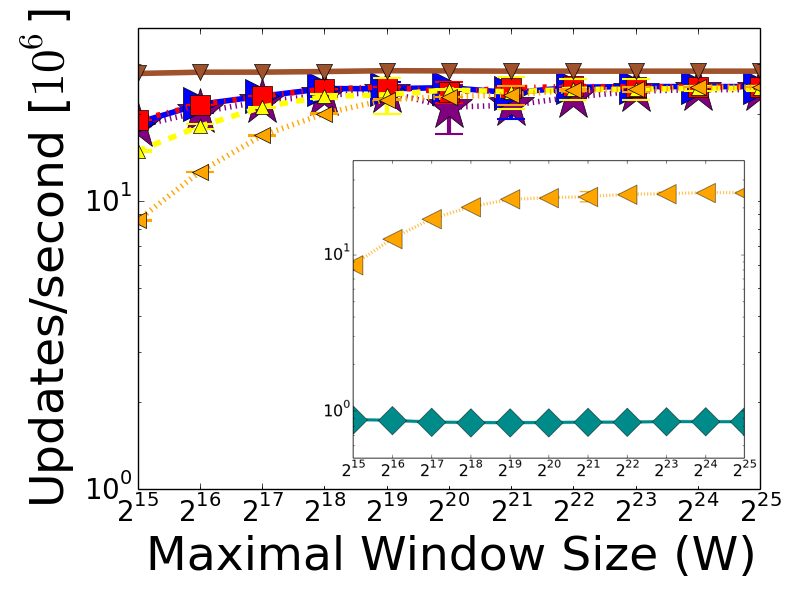}}\\
	\end{tabular}
%\vspace{-0.2cm}
%	\vspace{-0.12cm}\center{\includegraphics[width=1\columnwidth]{IFQGraphs/epsilon/legend.png}}
	\caption{Update operation runtime comparison as a function of the accuracy guarantee ($\epsilon$) and the maximum window size ($W$). Since the update speed of RAW is slower than the other algorithms, we have placed it in figure \ref{fig:row_update}. Each plot contains a subplot which compares ECM with our slowest algorithm, $\mathit{ACC_8}$ in this case. We see that ECM is much slower than $\mathit{ACC_8}$. We also compare our algorithms with $\mathit{WCSS}$ which is the state of the art for the simpler problem \mbox{of a \emph{fixed} sliding window.}}
	\label{fig:update}
\end{figure*}

\begin{figure*}[t]
	\center{
		\begin{tabular}{cc}
			\subfloat{\includegraphics[width=\matrixCellWidth]{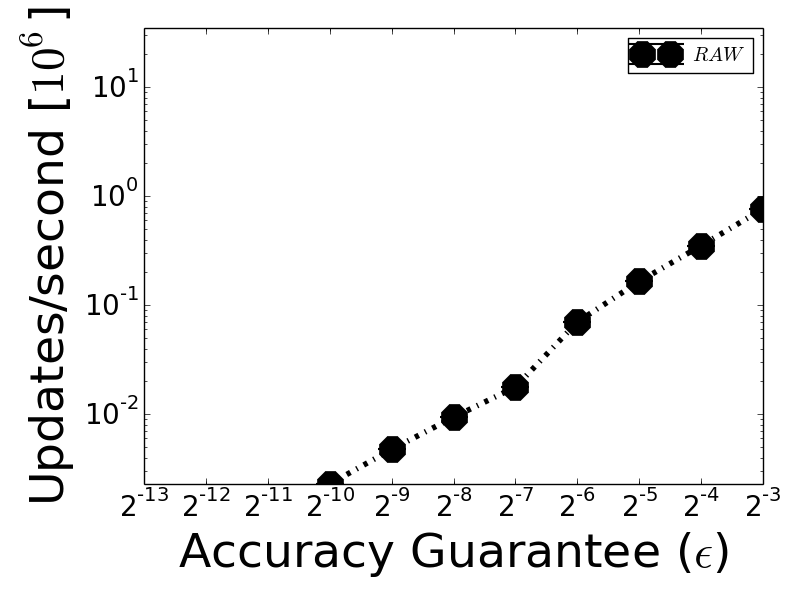}}
		\end{tabular}
	}
	\caption{Query operation runtime of $RAW$ algorithm comparison as a function of the accuracy guarantee ($\epsilon$) }
	\label{fig:row_update}
\end{figure*}
\fi

\ifdefined\VLDBREVISION
\begin{figure*}[t]
	\center{
\ifdefined\NINEPAGES		
		\vspace*{-0.4cm}
\fi
		\begin{tabular}{cc}
			\subfloat{\includegraphics[width=\matrixCellWidth]
				{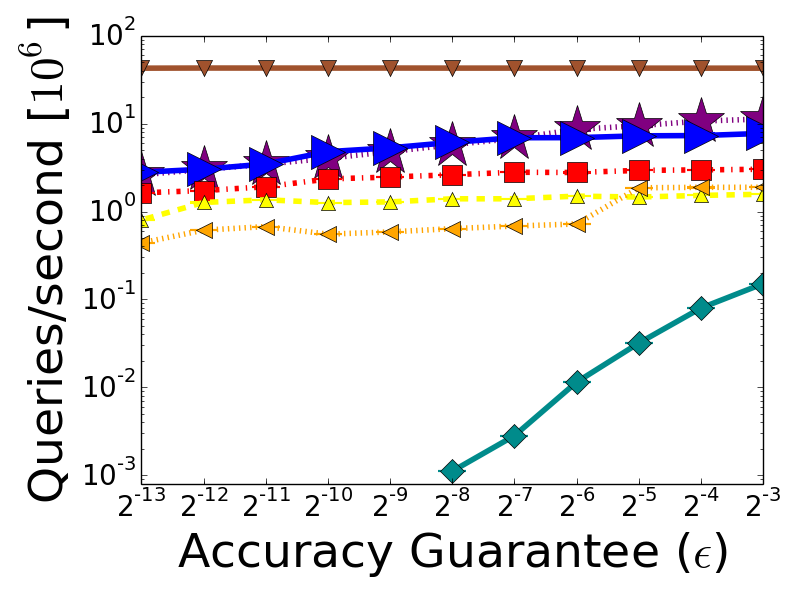}} &
			\subfloat{\includegraphics[width=\matrixCellWidth]{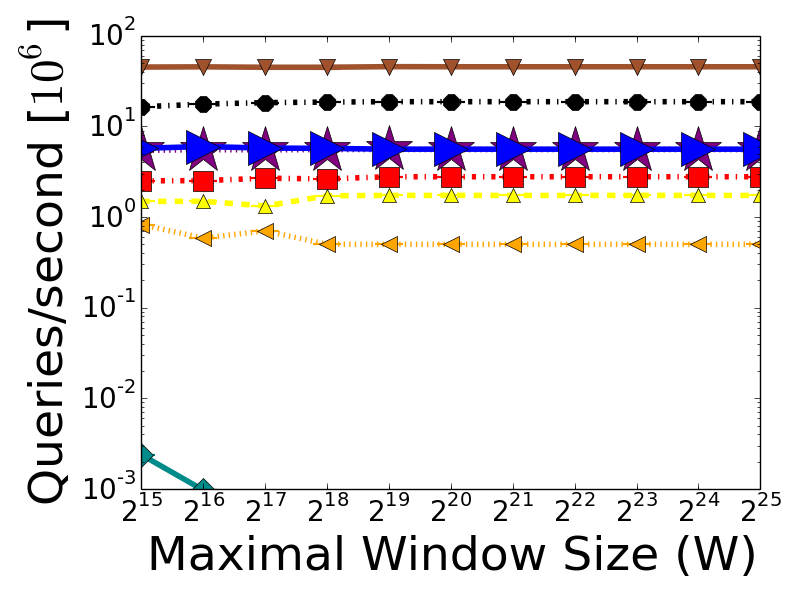}
				\ifdefined\NINEPAGES
	\vspace*{-3mm}
	\fi}
			\tabularnewline
			
            \multicolumn{2}{c}{\subfloat{\includegraphics[width = 12cm]
		{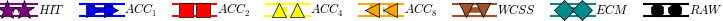}}}		
	\ifdefined\NINEPAGES
	\vspace*{-3mm}
	\fi
		\end{tabular}
		}
	\caption{\changed{Query operation runtime comparison as a function of the accuracy guarantee ($\epsilon$) and the maximum window size ($W$). When varying $\epsilon$, again, we managed to run $\mathit{RAW}$ only up to $\epsilon = 2^{-10}$ due to its memory consumption limitation. For graphs exploring varying $W$, we ran $\mathit{ECM}$ only up to $2^{16}$ due to time limitation. We compared our algorithms with $\mathit{WCSS}$ since it is the state of the art for the simpler problem of a \emph{fixed} sliding window but it can only answer queries \mbox{with fixed window size.}}
		\ifdefined\NINEPAGES
	\vspace*{-3mm}
	\fi}
	\label{fig:query}
\end{figure*}
\else
\begin{figure*}[t]
	\begin{tabular}{ccc}
		\subfloat[\changed{Backbone}]{\includegraphics[width=\matrixCellWidth]
			{IFQGraphs/epsilon/Query/NChicago16all_Queries_3-14all_epsilon.png}} &
		\subfloat[Datacenter]{\includegraphics[width = \matrixCellWidth]
			{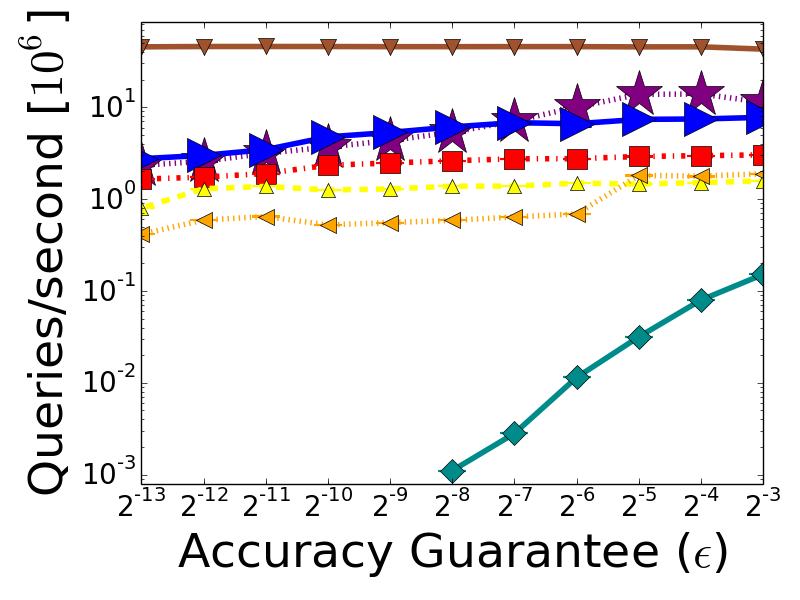}} &
		\subfloat[Edge]{\includegraphics[width = \matrixCellWidth]
			{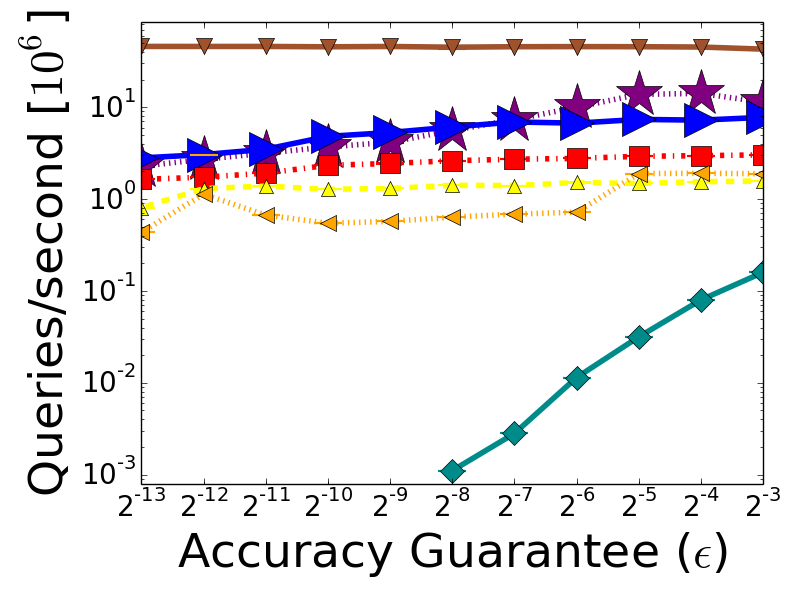}}
\ifdefined\NINEPAGES		
		\\ 	[-4ex]
\else
\\
\fi
\tabularnewline
\multicolumn{3}{c}{\subfloat{\includegraphics[width = 15cm]
		{IFQGraphs/epsilon/legend_w_raw.png}}}		
\ifdefined\NINEPAGES		
\\ 	[-2ex]
\else
\\
\fi
\addtocounter{subfigure}{-1}		
		\subfloat[Backbone]{\includegraphics[width=\matrixCellWidth]
			{IFQGraphs/window/Query/NChicago16_updates_15-25all_window.png}} &
		\subfloat[Datacenter]{\includegraphics[width = \matrixCellWidth]
			{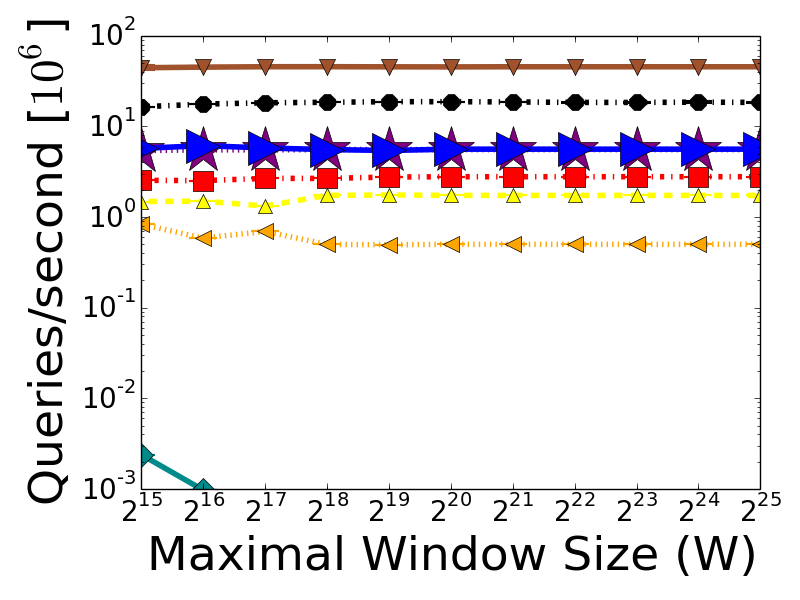}} &
		\subfloat[Edge]{\includegraphics[width = \matrixCellWidth]
			{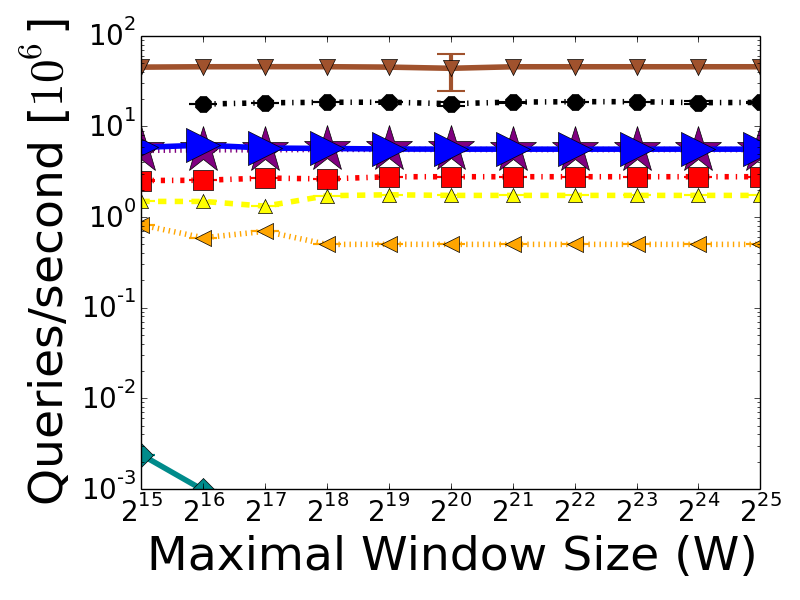}}
	\end{tabular}
%	\vspace{-0.12cm}\center{\includegraphics[width=1\columnwidth]{IFQGraphs/epsilon/legend_w_raw.png}}		
	\caption{Query operation runtime comparison as a function of the accuracy guarantee ($\epsilon$) and the maximum window size ($W$). For graphs as function of $\epsilon$, we managed to run $\mathit{RAW}$ only up to $\epsilon = 2^{-8}$ due to its memory consumption limitation. For graphs as function of $W$,we run $\mathit{ECM}$ only up to $2^{16}$ due to time limitation.We compare our algorithms with $\mathit{WCSS}$ since it is the state of the art for the simpler problem of a \emph{fixed} sliding window but it can only answer queries \mbox{with fixed window size.}}
	\label{fig:query}
\end{figure*}
\fi

\begin{comment}
\begin{figure*}[t]
	\center{
\ifdefined\NINEPAGES		
		\vspace*{-0.4cm}
\fi
		\begin{tabular}{cc}
			\subfloat{\includegraphics[width=\matrixCellWidth]
				{IFQGraphs/epsilon/Update/NChicago16raw_Updates_3-14raw_epsilon.png}} &
			\subfloat{\includegraphics[width=\matrixCellWidth]{IFQGraphs/memory/memory.png}}
		\end{tabular}
		}
	\caption{Update runtime of RAW algorithm and algorithms space comparison as function of the accuracy guarantee ($\epsilon$).}
	\label{fig:raw_mem}
\end{figure*}
\end{comment}

\begin{figure*}[t]
	\center{
\ifdefined\NINEPAGES		
		\vspace*{-0.4cm}
\fi
		\begin{tabular}{cc}
			\subfloat{\includegraphics[width=\matrixCellWidth]
				{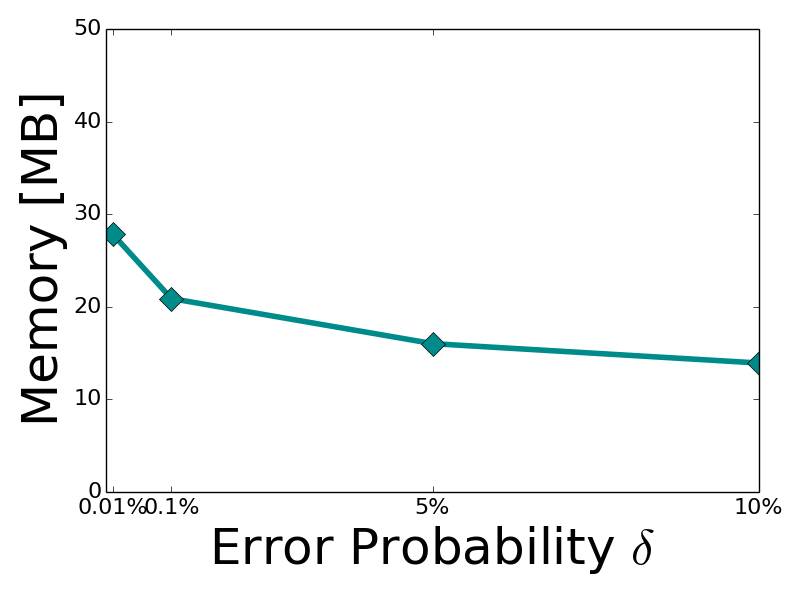}} &
			\subfloat{\includegraphics[width=\matrixCellWidth]{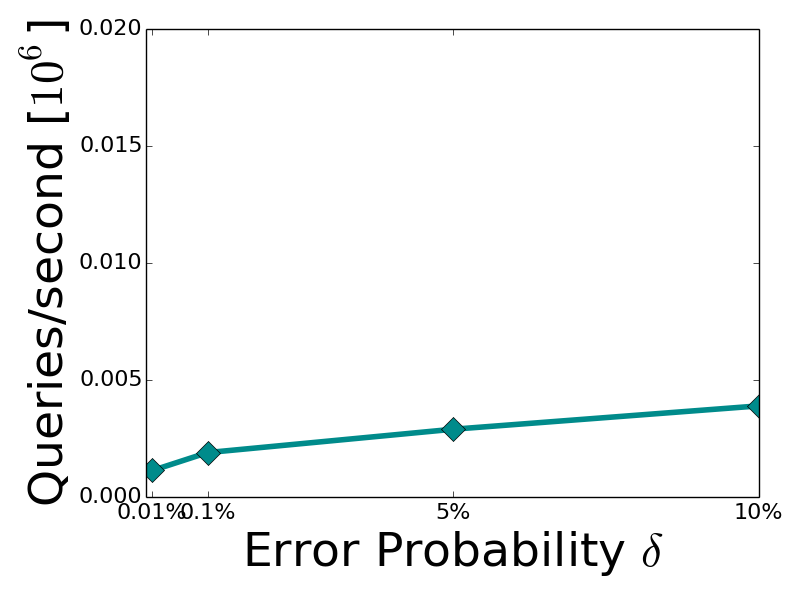}}
			\subfloat{\includegraphics[width=\matrixCellWidth]{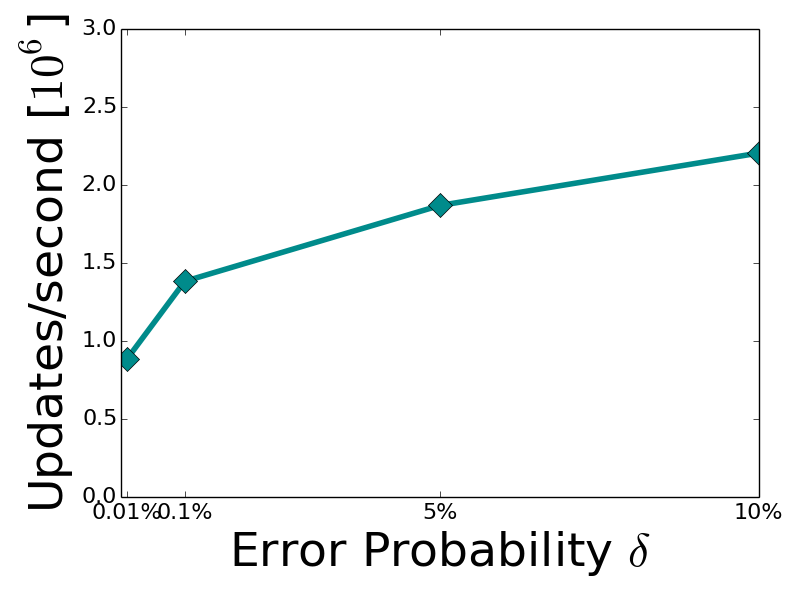}}
		\end{tabular}
		}
	\caption{\changed{$\mathit{ECM}$ space and performance comparison as functions of the error probability $\delta$ using \emph{Backbone} dataset, $\eps = 2^{-8}$ and window of size $2^{20}$. Note that the $y$-axes of these graphs is in linear scales.}}
	\label{fig:ecm_delta}
\end{figure*}

\begin{figure*}[t]
	\center{
\ifdefined\NINEPAGES		
		\vspace*{-0.4cm}
\fi

		\begin{tabular}{cc}
			\subfloat[Vary Interval sizes]{\includegraphics[width=\matrixCellWidth]{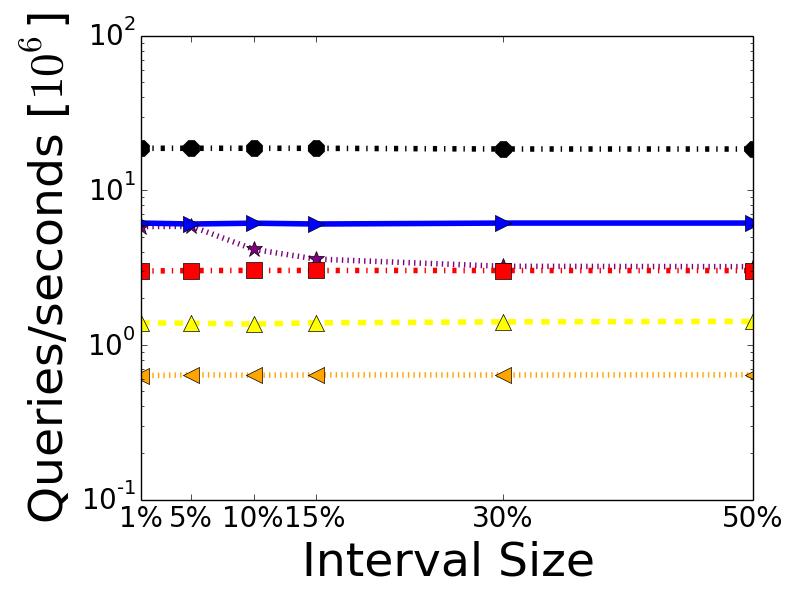}\label{vary_intervals}} &
			\subfloat[Observed Error]{\includegraphics[width=\matrixCellWidth]{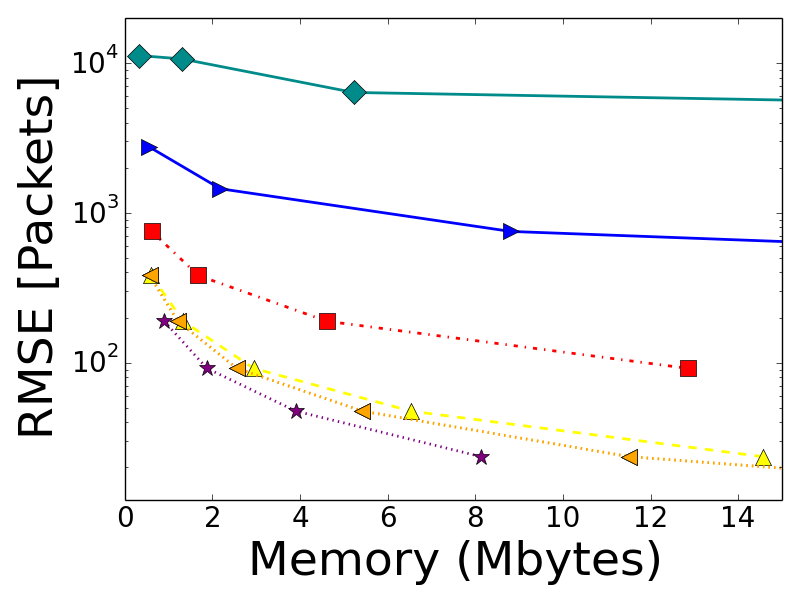}\label{emp_error}}
			\subfloat[Memory Consumption]{\includegraphics[width=\matrixCellWidth]{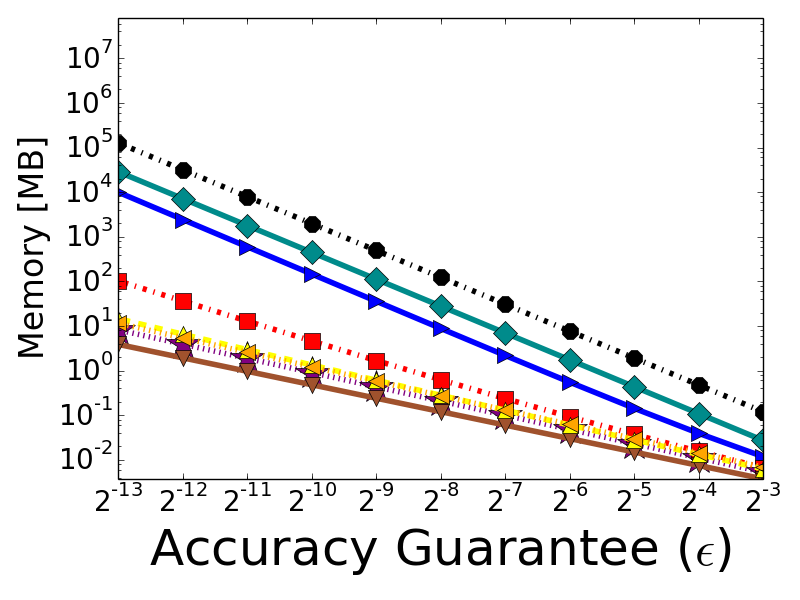}\label{memory}}
			\tabularnewline
            \multicolumn{2}{c}{\subfloat{\includegraphics[width = 12cm]
		{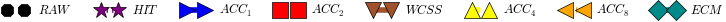}}}	
		\ifdefined\NINEPAGES
\vspace*{-3mm}
\fi
		\end{tabular}

		}
	\caption{\changed{(a) Query operation runtime comparison as a function of interval size (b) Root Mean Square Error comparison as a function of required memory and maximum window of size $2^{20}$ (c) Algorithms space comparison as a function of the accuracy guarantee ($\epsilon$).}
	\ifdefined\NINEPAGES
	\vspace*{-3mm}
	\fi
	}
	\label{fig:vary_intervals}
\end{figure*}

%\vspace*{-6mm}
\section{Evaluation}
\label{sec:eval}
We developed a C++ prototype of all algorithms described in this work: $\mathit{HIT}$, $\mathit{RAW}$, and instantiations of the $\mathit{ACC_k}$ protocols for $k=1,2,4,8$.
Here, the $\mathit{HIT}$ and $\mathit{ACC_k}$ algorithms are implemented using \changed{Space Saving~\cite{SpaceSavings}} as a building block.
Besides, we also implemented $\mathit{ECM}$-$\mathit{Sketch}$ ~\cite{papapetrou2015sketching} (a.k.a $\mathit{ECM}$) in C++ for comparison because the authors' code is in Java.
$\mathit{ECM}$ \changed{was} configured for error probability $\delta=0.01$\%.
\changed{As Table~\ref{tbl:comparison} shows, $\delta$ affects the space and performance of ECM. Specifically, the memory, update time, and query time, all logarithmically depend on $\delta^{-1}$. While the actual value of $\delta$ is application dependent, for performing a drill-down query (which translates into muliple interval queries), one may need $\delta$ to be quite small so that the overall error probability will be acceptable. Figure~\ref{fig:ecm_delta} shows $\mathit{ECM}$ space consumption and update time as functions of $\delta$. As expected, as $\delta$ increases, update and query operations become faster and $\mathit{ECM}$ consumes less space but the overall error probability is higher.}
We also \changed{compared} with the $\mathit{WCSS}$ algorithm~\cite{WCSS} as a general baseline since it is the state of the art for the more straightforward problem of a \emph{fixed} sliding window. Here again, we implemented $\mathit{WCSS}$ in C++ as its authors implemented it in Java.
For each algorithm, we \changed{evaluated} the speed of executing \IFQ{}
${(x,i,j)}$ and \text{ADD}${(x)}$ operations, as well as its memory requirements.

%We evaluated their performance and memory requirements for executing \text{ADD}${(x)}$ and\\ \IFQ{}${(x,i,j)}$ operations.
%In particular, for the HIT and $\mathit{ACC_k}$ algorithms, this includes the reduction of Section \ref{sec:reduction} using the WCSS algorithm~\cite{WCSS}.
%We compare our algorithms with the existing work, ECM sketch  (configured for error probability $\delta=0.01$\% and all algorithms compared with the (fixed window) WCSS algorithm~\cite{WCSS}, which serves as a baseline.

%In this section, we evaluate \text{ADD}${(x)}$ and \IFQ{}${(x,i,j)}$ operations speed, comparing between our C++ prototypes of the algorithms that solve \SIProblem{} problem: ACC$_1$, ACC$_2$, ACC$_4$, ACC$_8$, HIT, and RAW.
%In order to compare \IFQ{} operations, we apply the reduction of Section \ref{sec:reduction} with WCSS algorithm~\cite{WCSS} implementation to ACC$_1$, ACC$_2$, ACC$_4$, ACC$_8$ and HIT.
%In addition, we compare these algorithms with the (fixed window) WCSS algorithm~\cite{WCSS} for update and (frequency) query operations, which serves as a baseline.

The evaluation \changed{was} performed on an Intel(R) 3.20GHz Xeon(R) CPU E5-2667 v4 running Linux with kernel 4.4.0-71.
Each data point in all runtime measurements is shown as a 95\% confidence interval of 10 runs.
\ifdefined\EXTENDED
\subsection{Datasets}
\fi
Our evaluation includes \changed{a \emph{Backbone} dataset collected during 2016 from the backbone router `equinix-chicago'~\cite{CAIDACH16}. In the full version~\cite{full-version} we show these graphs with two additional packet traces (a data-center and an edge router) with very similar results}. 

\begin{comment}
the following datasets:
\begin{enumerate}
\ifdefined\EXTENDED
	\item A CAIDA 2016 backbone Internet trace~\cite{CAIDACH16}, denoted \emph{Backbone}.
	It consists of 97M packets, which are a mix of UDP/TCP/ICMP packets collected from the `equinix-chicago' high-speed monitor, a major backbone router.
\else
	\item \emph{Backbone} -- collected during 2016 from the backbone router `equinix-chicago'~\cite{CAIDACH16}.\vspace{1mm}
\fi	
\ifdefined\EXTENDED
	\item University campus data centers traces, denoted \emph{Datacenter}~\cite{Benson2010DC}.
	%and \emph{Univ\_2}.
	These data centers are located in the western/mid-western U.S. and are hosted on the premises of the organizations
	to serve local users. The \emph{Datacenter} dataset consists of 17M packets.
\else
	\item \emph{Datacenter} -- a data center located in the western U.S. and serves local users~\cite{Benson2010DC}.\vspace{1mm}
\fi	
\ifdefined\EXTENDED
	\item A trace from the border router of the Computer Science Department, University of California Los Angeles, denoted \emph{Edge}
	 %and \emph{UCLA\_UDP}.
	 consisting of 16M packets~\cite{UCLA-traces}.
\else
	\item \emph{Edge} -- a trace from the border router of the Computer Science department at UCLA~\cite{UCLA-traces}.\vspace{1mm}
\fi	
	 % and 18M packets respectively.
%    These packet traces were collected during August 2001.
\end{enumerate}
%A summary of key characteristics for workloads used in this work is given in Table~\ref{tbl:traces}.
\end{comment}
\subsection{Update Speed Comparison}
\label{sec:update_op}\vspace{3mm}
Figure~\ref{fig:update} compares the update speed. We start by exploring the trade-off of $\epsilon$ parameter with a fixed maximal window of size $W=2^{20}$.
Then, we explain the trade-off of window size parameter with a fixed $\epsilon = 2^{-8}$.
\subsubsection{Effect of $\epsilon$ on \changed{Update Time}}
Throughout, as $\epsilon$ decreases, more tables must be updated on every overflowed element.
Thus, update operations become slower when $\epsilon$ decreases.
%
%We conducted 10 runs of each dataset with varying $\epsilon$'s values, represent them as confidence interval in the graphs.
%
%As can be observed,
As depicted, $\mathit{HIT}$ update performance is close to $ACC_1$ and $ACC_2$.
As $k$ increases, there are more tables to update on every overflowed element, so the performance decreases.
This difference becomes especially noticeable with small $\epsilon$ values.

Recall that every update operation in $\mathit{RAW}$ means $4\epsilon^{-1}$ add operations, one for every $4\epsilon^{-1}$ instances of the $\mathbb A(\cdot, \eps/4)$
\ifdefined\EXTENDED
algorithm, which is $\mathit{WCSS}$ in our implementation.
\else
%algorithm.
algorithm, which is $\mathit{WCSS}$ in our implementation.
\fi
So, as $\epsilon$ decreases, update operations take more time.
Since the update speed of RAW is orders of magnitudes slower than the other algorithms, we have placed it in a separate \changed{graph} in which we managed to run this only for $\epsilon \ge 2^{-10}$ due to space limitation on the server.
This echoes Table~\ref{tbl:comparison}, which presents the analytical performance summary of the different algorithms.
Among our algorithms depicted in Figure~\ref{fig:update}, the slowest one is $\mathit{ACC_8}$,
as can be seen in the inner graphs, even $\mathit{ACC_8}$ processes items $57$-$210$ times faster than $\mathit{ECM}$.
\subsubsection{Effect of Window Size on \changed{Update Time}}
Figure~\ref{fig:update} shows also the effect of window size when $\epsilon$ is fixed to $2^{-8}$.
All algorithms perform better when the window size is larger
%because as the window size decreases, the probability to initialize a new block increases, which causes calculations of corresponding tables.
as this means fewer blocks and table accesses.
The $\mathit{ACC_k}$ algorithms get slower as $k$ increases as they need to update more tables.
Again, we compared the most inefficient algorithm $\mathit{ACC_8}$ with $\mathit{ECM}$ in the inner graphs;
$\mathit{ACC_8}$ processes items $50$-$218$ times faster than $\mathit{ECM}$ for the given $\epsilon$~values.
%We show that our slowest algorithm is more efficient than the existing work, $\mathit{ECM}$, in both cases of fixed sized window size to $2^{20}$ with $\epsilon$ in the range $2^{-4}$-$2^{-13}$ and in the case of fixed sized $\epsilon$  to $2^{-8}$ with window size in the range $2^{13}$-$2^{19}$.
%For the $\mathit{ACC_k}$ algorithms, as $k$ increases, the number of tables to update grows up on every overflowed element which leads to inferior performance.

\subsection{Query Speed Comparison}
\label{sec:query_op}
For query speed comparison, we \changed{chose} random intervals, each of size $1$\% of the total window's size.
%The comparison is done between $\mathit{HIT}$, $\mathit{RAW}$, $ACC_1$,  $ACC_2$, $ACC_4$ and $ACC_8$ for interval queries, and WCSS for point queries.
We begin the evaluation by exploring the impact of the $\epsilon$ parameter with a fixed window of size $2^{20}$.
Then, we explain the trade-off of the window size parameter with fixed $\epsilon = 2^{-8}$.
The performance of the improved algorithms is compared with the existing work, $\mathit{ECM}$, and $\mathit{WCSS}$, recall that $\mathit{WCSS}$ can only answer queries with fixed window size.
\subsubsection{Effect of $\epsilon$ on \changed{Query Time}}
%The comparison of interval query is given in Figure~\ref{fig:query}.
%Here also for varying $\epsilon$'s values we conducted 10 runs on each dataset, represented as confidence interval.
%As can be observed, $\mathit{RAW}$ is the most efficient because each interval query is translated to two WCSS queries.
As shown in Figure~\ref{fig:query}, $\mathit{RAW}$ is the fastest as each interval query is translated to two $\mathit{WCSS}$ queries.
We managed to run $\mathit{RAW}$ only up to $\epsilon = 2^{-10}$ due to its memory consumption limitation (see section \ref{sec:memory}).

$\mathit{HIT}$ computes any block interval frequency by using the hierarchical tree tables, greedily choosing the highest possible level each time.
For decreasing $\epsilon$ values, the blocks numbers increases, so the queried interval crosses more blocks and accesses more tables.
Consequently, we \changed{got} slower interval query operations.

For the $\mathit{ACC_k}$ algorithms, for increasing $k$ values we get fewer queries per seconds as we read more tables on average.
For example, $\mathit{ACC_1}$ computes any block frequency by querying at most $3$ tables, while $\mathit{ACC_2}$ does the same by accessing no more than
\ifdefined\EXTENDED
$5$ tables as explained in Section~\ref{sec:acc}.
\else
%$5$.
$5$ tables as explained in Section~\ref{sec:acc}.
\fi
Query operations runtime depends also on the interval itself;
there are ``good'' intervals in which the corresponding blocks have table at level $k - 1$, so one table access for each is sufficient.
Therefore, we \changed{chose} random intervals for every query.
As $\epsilon$ value decreases, block sizes become smaller and the number of tables grow. % which leads to maintaining more tables.
In this case, not all the tables fit in memory and we \changed{experienced} paging that causes lower query performance.
Recall that while $\mathit{WCSS}$ is the fastest, it solves the much simpler problem of a fixed window size and only serves as a best case reference~point.
%\textbf{Ran: we cannot avoid discussing WCSS here. Yes, it is faster, but solves a simpler problem. It's not good to leave it for the reviewer to understand this.}
$\mathit{ECM}$ answers queries in a very inefficient way compared to our algorithms.
We only run $\mathit{ECM}$ up to \changed{$\epsilon = 2^{-8}$} due to time limitation.
As expected, its performance decreases for decreasing $\epsilon$ values.

\subsubsection{Effect of Window Size on \changed{Query Time}}
As mentioned before, we \changed{evaluated} queries by choosing random intervals of size 1\% of window's size, when $\epsilon$ is fixed to $2^{-8}$.
Figure~\ref{fig:query} shows that all algorithms' query performance is not very sensitive to the window size.
This is because the number of tables accessed depends on the ratio between the interval and window sizes.
We \changed{ran} $\mathit{ECM}$ only up to \changed{$2^{16}$} due to time limitation.
The performances of our algorithms are orders of magnitudes better than $\mathit{ECM}$ also in this case.

\subsection{Memory Consumption Comparison}
\label{sec:memory}
Figure~\ref{memory} shows the space consumed by our algorithms as well as $\mathit{ECM}$
%when equipped with decreasing $\epsilon$ values.
for a given $\epsilon$ value.
As seen, the smaller $\epsilon$ gets, all algorithms consume more space.
We can see that $\mathit{ECM}$ is more compact than $\mathit{RAW}$ but consumes more space than the others.
As mentioned before, RAW maintains $4\oneOverE$ separate $\mathit{WCSS}$ instances so its space consumption is the largest.
For the $\mathit{ACC_k}$ algorithms, as $k$ increases, the overall number of tables entries for overflowed elements decreases resulting is better space consumption.
So there is a trade-off between the speed and required spaces by adjusting the parameter $k$
, using Figure~\ref{memory} and figures \ref{fig:update},\ref{fig:query} can help choosing $k$ parameter according to the desired speed and memory consumption.
Yet, $\mathit{ECM}$ consumes more space than $ACC_1$ which has the highest memory consumption among the $\mathit{ACC_k}$ algorithms family.
As shown, the memory consumption of $\mathit{ACC_8}$ is close to $\mathit{HIT}$.
Yet, $\mathit{HIT}$ is the most efficient among the algorithms that solve \SIProblem{} because its data structure is the most compact \changed{but its query performance affected by the interval size as explained in section \ref{sec:intervals} so for larges interval sizes we may prefer $ACC$ algorithm over $\mathit{HIT}$}.
%the overhead from its added data structure is the lowest.
Recall that while $\mathit{WCSS}$ is the most compact algorithm in term of space, it solves the much simpler problem of a fixed window size and only serves as a best case reference~point.
\changed{
\subsection{Interval Size Comparison}
\label{sec:intervals}
Figure~\ref{vary_intervals} shows query runtime performance of our algorithms as a function of interval size.
The query operation performance was measured with random intervals of varying sizes: $1$\%, $5$\%, $10$\%, $15$\%, $30$\%, or $50$\% of the total window's size while fixing $\eps=2^{-8}$ and $W=2^{20}$.
}

\changed{
As expected, the query performance of $\mathit{HIT}$ gets slower as the size of the interval gets larger, because the queried interval crosses more blocks and accesses more tables.
In contrast, query performance of $\mathit{ACC_k}$ algorithms is not affected by interval size as $\mathit{ACC_k}$ algorithms consider only the edges of the queried interval. That is, when the edges are $i$ and $j$, they compute the frequency of the given item from the beginning of the frame till $\mathit{block}_i$ and $\mathit{block}_j$ and subtract the results.
$\mathit{RAW}$ algorithm query performance is not affected by interval size since it is translated to two $\mathit{WCSS}$ queries regardless of interval size.
$\mathit{ECM}$ algorithm is not included in the graph since it is orders of magnitudes slower than our algorithms so we will not see the difference between them (see figure \ref{fig:query}). $\mathit{ECM}$ query performance is affected by interval size since its query operation depends on Exponential Histograms~\cite{DatarGIM02} query. As interval size gets larger, Exponential Histograms scans a larger sequence of \emph{buckets} and as a result, $\mathit{ECM}$ query gets slower.}

\changed{In conclusion, when the interval size is big, we would prefer to choose $\mathit{ACC_1}$ over $\mathit{HIT}$ when there is sufficient memory. We expect that as data rates and volumes get higher, one would use smaller $\eps$ values, making $\mathit{ACC_k}$ increasingly more attractive also for larger $k$ values.}
\changed{\subsection{Root-Mean-Square Error Comparison}
\label{sec:emp_error} Figure~\ref{emp_error} shows the empirical Root Mean Square Error ($\mathit{RMSE}$) in correlation with the required memory for $\mathit{ACC_k}$ algorithms, $\mathit{HIT}$ and $\mathit{ECM}$ with window of size $2^{20}$.}
\changed{The observed errors are lower than the user-selected value $\eps$. 
Since $\mathit{ACC_k}$ algorithms and $\mathit{HIT}$ solve the same $n$-Interval instance, their empirical error is equal for same $\eps$ values. 
The difference between the algorithms comes from the memory requirements which differ for the same $n$ value.
In general, a lower space consumption required for a specific $\eps$ value translates into better empirical error.
For example, $\mathit{ACC_1}$ consumes more memory than $\mathit{ACC_2}$ for the same $\eps$. Thus, for a given memory budget, $\mathit{ACC_2}$ is more accurate than $\mathit{ACC_1}$ and $\mathit{HIT}$ is more accurate than both. $\mathit{ECM}$ was measured with error probability $\delta=0.01$\% which led to large memory consumption relative to $\mathit{HIT}$ and $\mathit{ACC_k}$ algorithms; as a result its empirical error higher than others but yet lower than the theoretical value.}

%\section{Time Based Intervals}

\section{Extensions and Applications}
\label{sec:extensions}
%\ifdefined\NINEPAGES
Here, we briefly discuss how our solutions can be applied to temporal queries,  weighted stream, distributed stream, heavy hitters, and hierarchical heavy hitters (HHH).
\ifdefined\UNDEF
 We expand on both subjects in the full version of the paper~\cite{full-version}.

First, we describe how to extend the framework to \emph{time intervals}. That is, consider a router that processes packets and wishes to find the most frequent flows between time $t_1$ and $t_2$, measured in \emph{seconds}. While the above algorithms are designed for packet-count intervals, they can be used for time intervals as well. To achieve this, we feed the number of packets that arrived in each second into a \emph{Sliding Ranker}~\cite{slidingRanker}. The Sliding Ranker allows us to estimate the number of packets that arrived in a certain time interval (with a small additive error). Given that packet-count estimate, we can query our algorithm and return an answer with respect to the time interval.

Next, we describe how our algorithms can be used for answering interval HHH queries (see~\cite{HHHMitzenmacher} for formal definitions). In~\cite{HHHMitzenmacher}, Mitzenmacher et al. proposed combining their approach (which originally utilized Space Saving~\cite{SpaceSavings}) with sliding window algorithms such as~\cite{HungLT10,WCSS} to solve HHH on sliding windows. However, such an approach yields a fixed window size algorithm. By replacing the underlying black box algorithm by our interval query solutions, we get an algorithm that solves HHH on interval queries.
\else
\subsection{Time Based Intervals}
In this section, we describe how to extend our algorithms for supporting time intervals.
The idea is that sometimes what matters is the flow frequencies in a time interval rather than during a item-count interval.
For example, if we want to allow a user to make 100 queries/sec to an API, we need to measure the number of times this user has accessed the system during the last second.
In such a setting, we consider a \emph{timed stream} $\mathcal S = \langle x_1,t_1\rangle, \langle x_2,t_2\rangle,\ldots \in (\mathcal U \times N)^*$.
Here, each item has an integer valued timestamp and we assume that the items arrive in order, $i.e.$, $t_1\le t_2\le \ldots$.

We also assume that the number of elements that arrive in a single time-frame is bounded by $R\in \mathbb N$.
%This is often a reasonable assumption.
In practice, this is a reasonable assumption; for example, if we perform the measurement over a 1Gbps link, and each item must be of size of at least 64 bytes for its headers, then we can set $R\triangleq 10^{9}/(8\cdot 64) = 2M$ [items / second].
We denote by \xTimeWindowFrequency{} the frequency of $x$ within the last $t$ timestamps.
That is, if the query time is $T$, then $\xTimeWindowFrequency\triangleq |\set{\langle x,t_i\rangle\in \mathcal S \mid t_i \ge T-t}|$.
Similarly, we define the frequency within a time window as $\xTimeIntervalFrequency\triangleq \xTimeWindowFrequency[i] - \xTimeWindowFrequency[j]$.
The time based interval algorithms goal is then to answer the following queries:
\begin{itemize}
	\item \BTIFQ{}$\bm{(x,i,j)}$: given an element $x\in\mathcal U$ and indices $j\le i\le W$, return an estimate $\xTimeIntervalFrequencyEstimator$ of $\xTimeIntervalFrequency$.
\end{itemize}
Finally, if an algorithm's error is at most an $\eps$ fraction of the overall possible traffic in $\Tau$ time, we say that it solves the \TIFProblem{} problem.
%if its error is at most an $\eps$ fraction of the overall possible traffic in $\Tau$ time, $i.e.$,
That is, its estimation needs to satisfy
%Finally, we say that an algorithm  solves the \TIFProblem{} problem if its error is at most an $\eps$ fraction of the overall possible traffic in $\Tau$ time, $i.e.$, it satisfies
$$\forall j\le i\le \Tau: \xTimeIntervalFrequency\le \xTimeIntervalFrequencyEstimator \le \xTimeIntervalFrequency + \Tau\cdot R\cdot\eps.$$
Our construction has two parts:
We maintain a \ParameterizedIFProblem{\Tau\cdot R}{\eps/2} solution in addition to a data structure that translates time intervals into item intervals.
For this, we use Ben Basat's \emph{Sliding Ranker} (SR) algorithm~\cite{slidingRanker} that can compute a sliding window sum over an integer stream, where the size of the window is given at query time.
SR has parameters $\angles{R,\mathfrak W,\Delta}$; it processes a stream in $\frange{R}$ such that upon a query for some $i\le \mathfrak W$, it computes a $\Delta$-additive approximation for the sum of the last $i$ elements.
Every timestamp, we feed the \emph{number of items} that arrived into an SR with parameters $\angles{R,\Tau,\Tau\cdot R\eps/2}$.
Given a time-interval query $x,i,j$, we use the SR for computing the number of items sent since time $i$ and from time $j$.
We then use these estimations to query the \IFText{} instance for the estimated item-interval.
Since the SR and \IFText{} each has an error of $\Tau\cdot R\eps/2$, we satisfy the error guarantee.
The memory consumption of SR for $\Delta = \Theta(R\mathfrak W)$ is just $O(R\mathfrak W\Delta+\log \mathfrak W)=O(\oneOverE+\log \mathfrak W)$ bits.

\subsection{Supporting Heavy-Hitters}
We now show how one can use the described algorithms to support heavy hitters queries over a given interval. 
Denote by $\IntervalHH\triangleq\IHHDef$ the set of heavy hitters items that appeared at least a $\theta$ fraction of the queried interval for given integers $i\le j\le W$ and a real number $\theta\in[0,1]$.
 \BIHHQ{}$\bm{(\theta,i,j)}$ operation returns an estimate $\IntervalHHEstimator\subseteq\mathcal U$ that approximates $\IntervalHH$ given indices $i\le j\le W$.
The algorithms solves \BIHHQ{}$\bm{(\theta,i,j)}$ and guarantees  $$\IntervalHH\subseteq \IntervalHHEstimator \subseteq \set{x\in\mathcal U\mid \xIntervalFrequency\ge \intervalHHThreshold - \epsError}.$$
That is, the estimated set must contain all elements that appear at least a $\theta$ fraction of the interval and must not have any members whose frequency is lower than $\theta\cdotpa{j-i} - \epsError$.
Given that the described algorithms solve the \IFProblem{} problem, by the following observation they also solve the \IHHProblem{} problem.

\begin{observation}
	Any algorithm $\mathbb A$ that solves \IFProblemD{} can answer an \IHHQ{} by returning
	\ifdefined\EXTENDED
	$$\IntervalHHEstimator\triangleq \set{x\in\mathcal U\mid \xIntervalFrequencyEstimator\ge\intervalHHThreshold}.$$
	\else
	$\IntervalHHEstimator\triangleq \{x\in\mathcal U\mid \xIntervalFrequencyEstimator\ge\intervalHHThreshold\}.$
	\fi
\end{observation}
%consider the problem of \emph{Interval Window Frequency estimation}.
Specifically, all algorithms presented in this paper can compute the set $\IntervalHHEstimator$ in time $O(\oneOverE)$ without iterating over all universe elements.
Thus, we note that all proposed algorithms can efficiently compute the $\IntervalHHEstimator$ \mbox{suggested above.}

\subsection{Hierarchical Heavy Hitters}

Next, we describe how our algorithms can be used for answering interval HHH queries (see~\cite{HHHMitzenmacher} for formal definitions). In~\cite{HHHMitzenmacher}, Mitzenmacher et al. proposed combining their approach (which originally utilized Space Saving~\cite{SpaceSavings}) with sliding window algorithms such as~\cite{HungLT10,WCSS} to solve HHH on sliding windows. However, such an approach yields a fixed window size algorithm. By replacing the underlying black box algorithm by our interval query solutions, we get an algorithm that solves HHH on interval queries. This is also orthogonal to the other approaches; a combination of the extensions proposed in this chapter would allow finding HHH over time-based intervals, finding distributed HHH (see the following section), or finding HHH in terms of traffic volume (Section~\ref{sec:volumetric}).

\fi

\subsection{The Distributed Model}
\label{sec:distributed}

We now consider applying our algorithms in distributed settings.
Here, multiple streams are received at various sites $S_1$,\dots,$S_r$ ($r>1$) and each site maintains its own instance of the chosen algorithm, e.g., $ACC_k$, $\mathit{HIT}$, etc.
Obtaining a global view of the system's status requires merging data structures from all individual sites.
The common way of serving such queries is to have all individual sites transmit a copy of their data structures to a central controller $C$, which merges them into a global data structure.
This can be done either periodically assuming synchronized clocks between the sites, or in a coordinated manner initiated periodically by $C$.
Queries are forwarded to the controller that computes the reply based on its merged data-structure.
This model is communication efficient when queries are frequent, since queries are served directly by the controller and the rate in which the distributed sites need to communicate with the controller can \mbox{be lower than the query rate.}

%, each site observes a local single or many streams.

\begin{comment}
While the traditional model is centralized, which requires merging of separate instances of each suggested algorithm data structures.
In this model each of the sites has a communication channel with special site, $C$, whose goal is answer range queries, but other sites cannot communicate with each other.A site holds the data structure of the algorithm, $C$ holds a combination of these data structures.\\
Upon arrival of item $x$ to site $S_i$, $S_i$ may decides to send $C$ an update of its structure.In this case, $C$, trigger a communication from the other sites and holds an updated image of the distributed streams. A range query is passed directly to $C$ and treated in the same way as one site model \ref{sec:hit} and \ref{sec:acc}.The cons of this model is the total number of communications among all the sites.
\end{comment}

In contrast, when queries are not as frequent, the above solution is inefficient, since the sites needlessly update the controller.
To that end, by applying the time based intervals adaptation, our solution enables the reverse model.
That is, given a range query, it is directly propagated to each of the $r$ distributed sites.
Each site returns its locally computed portion and all replies are then merged into a global one.
The reason why time based intervals are needed is that individual streams might arrive at different rates to the various sites.
Hence, it is meaningless to merge the results of queries on an item based window or range.
For this reason, the above approach cannot be applied to item based sketch~algorithms.

Last, as mentioned before, our algorithms provide an $\epsilon$ error guarantee.
Hence, when each of the sites runs its independent instance, the overall error guarantee of the distributed model becomes $r\epsilon$.
Another way of looking at this is that since the space requirement is inversely proportional to the error guarantee, the space requirement for a given error grows with $r$.
Since usually $r$ is a small constant, for most systems this is acceptable.

\subsection{Supporting Traffic Volume Heavy-Hitters}\label{sec:volumetric}
It is often desired to find the heavy hitters in terms of traffic volume. That is, consider a stream in which its item has a \emph{size} and we wish to find the flows that account for most of the bandwidth in a given interval. Formally, we consider a \emph{weighted stream} $\mathcal S = \langle x_1,\mathfrak w_1\rangle, \langle x_2,\mathfrak w_2\rangle,\ldots \in (\mathcal U \times \{1,2,\ldots,M\})^*$ and define a flow's volume as the sum of sizes for items that belong to it.

%Intuitively, to address this problem we wish to replace the Space Saving instance in Algorithm~\ref{alg:reduction} with
Intuitively, to address this problem we can add the weight of the item in Line~\ref{line:ssAdd} of Algorithm~\ref{alg:reduction}, and change the condition of Line~\ref{line:SSoverflow} to consider whether the current estimation exceeds a new multiple of $\blockSize\cdot M$.
The Space Saving algorithm~\cite{SpaceSavings} can find weighted heavy hitters over a stream with $O(\log\epsilon^{-1})$ update time~\cite{berinde2010space}. Recent breakthroughs~\cite{DIM-SUM, FAST, anderson2017high} improve this runtime to a constant. Thus, we can solve the interval volume estimation and (weighted) heavy hitters problems with the same asymptotic complexity as the unweighted variants and with an error of at most $WM\epsilon$. This is a generalization of the result of~\cite{FAST} that finds weighted heavy hitters over \mbox{fixed size windows.}

\section{Discussion}
\label{sec:discussion}
In this paper, we studied the problems of flow frequency estimation over intervals that are passed at query time.
Such capabilities can be useful when one wishes to maintain the above statistics over multiple sliding windows and for performing drill-down queries, e.g., for root cause analysis of network anomalies.

We presented formal definitions of these generalized problems and explored three alternative solutions: a naive approach (\textit{RAW}) and more sophisticated solutions called \textit{HIT} and $ACC_k$.
Both \textit{HIT} and $ACC_k$ process updates in $O(1)$, but differ in their space vs. query time tradeoff:
\textit{HIT} is asymptotically memory optimal but answers queries in logarithmic time whereas $ACC_k$ processes queries in $O(1)$ but consumes more space.
\changed{Moreover, \textit{HIT} interval queries performance is affected by interval size: as interval size gets larger, query gets slower. In contrast, \changed{the} $ACC_k$ algorithms are not affected by interval size.}
In fact, \textit{HIT}'s space requirement is similar to the memory requirement of the state of the art algorithm that can only cope with fixed size windows.
Hence, \textit{HIT} is adequate when space is tight \changed{or the intervals are small} while $ACC_k$ is suitable for real time query processing.
Both our advanced algorithms are faster and more space efficient than ECM~\cite{papapetrou2015sketching}, the previously known solution for interval queries.
This is true both asymptotically and in measurements over real-world traces, in which we demonstrated orders of magnitude runtime improvements as well as at least $40\%$ \mbox{memory reductions for similar estimation errors.}
\ifdefined\VLDB
Our approach can be applied to additional related problems.
For example, we showed in Section~\ref{sec:extensions} how to adapt our algorithms to answer queries over time based intervals as well as to identifying heavy hitters.
This can be further generalized to the \emph{hierarchical heavy hitters} (HHH) problem~\cite{HHHMitzenmacher}, which is useful in detecting distributed denial of service attacks (DDoS). In the latter, one can replace the Space Saving \changed{instances} employed by~\cite{HHHMitzenmacher} with $HIT$ or $ACC_k$ to detect HHH over query intervals!

%\vbox{
\noindent \textbf{Code Availability:} \changed{All code is available online~\cite{opensource}.}
\paragraph{Acknowledgements:} We thank the anonymous reviews for their insightful comments that helped improve this paper. We are also grateful for many comments and helpful observations made by Dimitrios Kaliakmanis.
\else
\vbox{
Our approach can be applied to additional related problems.
For example, we showed in Section~\ref{sec:extensions} how to adapt our algorithms to answer queries over time based intervals as well as to identifying heavy hitters.
This can be further generalized to the \emph{hierarchical heavy hitters} (HHH) problem~\cite{HHHMitzenmacher}, which is useful in detecting distributed denial of service attacks (DDoS). In the latter, one can replace the Space Saving \changed{instances} employed by~\cite{HHHMitzenmacher} with $HIT$ or $ACC_k$ to detect HHH over query intervals!\\\\
\textbf{Code Availability:} \changed{All code is available online~\cite{opensource}.}\\
\textbf{Acknowledgements:} We are grateful for many helpful comments and observations made by \mbox{Dimitrios Kaliakmanis.}
}
\fi

\newpage
\bibliographystyle{abbrv}
\bibliography{refs}  % refs.bib is the name of the Bibliography in this case

\end{document}